\documentclass[11pt,a4page]{llncs} 

\usepackage{enumitem,amsmath,amssymb}

\usepackage{tikz}
\usetikzlibrary{calc,positioning,arrows}
\def\centerarc[#1](#2)(#3:#4:#5)
    { \draw[#1] ($(#2)+({#5*cos(#3)},{#5*sin(#3)})$) arc (#3:#4:#5); }

\usepackage{graphicx,xcolor}
\usepackage{float}
\usepackage{wrapfig}

\usepackage{hyperref}
\hypersetup{
  colorlinks   = true, 
  urlcolor     = blue, 
  linkcolor    = blue, 
  citecolor   = red 
}

\usepackage[skip=.5ex]{parskip}
\newtheorem{fact}{Fact}

\newtheorem{rrule}{Reduction Rule}{\bf}{\it}
\newtheorem{brule}{Branching Rule}{\bf}{\it}

\newcommand{\NP}{\mathsf{NP}}

\def\MC{{\sc mc}} 
\def\PMC{{\sc pmc}} 
\def\3SAT{{\sc 3-sat}}
\def\NAE3SAT{{\sc positive nae 3-sat}}
\def\1IN3SAT{{\sc positive 1-in-3 3-sat}}
\def\P1IN3SAT{{\sc positive planar 1-in-3 3-sat}}

\begin{document}

\title{
The Perfect Matching Cut Problem Revisited
}

\author{Van Bang Le\inst{1} \and Jan Arne Telle\inst{2}
}
\institute{
Institut f\"{u}r Informatik, Universit\"{a}t Rostock,
Rostock, Germany\\ 
\email{van-bang.le@uni-rostock.de}
\and
Department of Informatics, University of Bergen,
N-5020 Bergen, Norway\\
\email{Jan.Arne.Telle@uib.no}
}

\maketitle
\thispagestyle{plain} 


\begin{abstract}
In a graph, a perfect matching cut is an edge cut that is a perfect matching. \textsc{perfect matching cut} (\PMC) is the problem of deciding whether  a given graph has a perfect matching cut, and is known to be $\NP$-complete. We revisit the problem and show that \PMC\ remains $\NP$-complete when restricted to bipartite graphs of maximum degree~3 and arbitrarily large girth. 
Complementing this hardness result, we give two graph classes in which \PMC\ is polynomial time solvable. The first one includes claw-free graphs and graphs without an induced path on five vertices, the second one properly contains all chordal graphs. 
Assuming the Exponential Time Hypothesis, we show there is no $O^*(2^{o(n)})$-time algorithm for \PMC\ even when restricted to $n$-vertex bipartite graphs, and also show that \PMC\ can be solved in $O^*(1.2721^n)$ time by means of an exact branching algorithm. 

\end{abstract}

\section{Introduction}
In a graph $G=(V,E)$, a \emph{cut} is a partition $V=X\cup Y$ of the vertex set into disjoint, non-empty sets $X$ and $Y$. The set of all edges in $G$ having an endvertex in~$X$ and the other endvertex in~$Y$, written $E(X,Y)$, is called the \emph{edge cut} of the cut $(X,Y)$. A \emph{matching cut} is an edge cut that is a (possibly empty) matching. 
Another way to define matching cuts is as follows; see~\cite{Chvatal84,Graham70}: a cut $(X,Y)$ is a matching cut if and only if each vertex in~$X$ has at most one neighbor in~$Y$ and each vertex in~$Y$ has at most one neighbor in~$X$. 
\textsc{matching cut} (\MC) is the problem of deciding if a given graph admits a matching cut and this problem has received much attention lately; see~\cite{ChenHLLP21,GolovachKKL21} for recent results. 

An interesting special case, where the edge cut $E(X,Y)$ is a \emph{perfect matching}, 
was considered in~\cite{HeggernesT98}. The authors proved that \textsc{perfect matching cut} (\PMC), the problem of deciding if a given graph admits an edge cut that is a perfect matching, is $\NP$-complete. 
A perfect matching cut $(X,Y)$ can be described as a $(\sigma, \rho)$ 2-partitioning problem~\cite{Prosk}, as every vertex in~$X$ 
must have exactly one neighbor in~$Y$ and every vertex in~$Y$ 
must have exactly one neighbor in~$X$. By results of~\cite{Bui,Prosk,Vatshelle} 
it can therefore be solved in FPT time when
parameterized by treewidth or cliquewidth (to mention only the two most
famous width parameters) and in XP time when parameterized by mim-width
(maximum induced matching-width) of a given decomposition of the graph.
For several classes of graphs, like interval and permutation, a
decomposition of bounded mim-width can be computed in polynomial-time~\cite{Belmonte}, 
thus the problem is polynomial on such classes. 

In this paper, we revisit the \PMC\ problem. Our results are:
\begin{itemize}
\item While \MC\ is polynomial time solvable when restricted to graphs of maximum degree~$3$ and its computational complexity is still open for graphs with large girth, we prove that \PMC\ is $\NP$-complete in the class of bipartite graphs of maximum degree~$3$ and arbitrarily large girth. 
Further, we show that \PMC\ cannot be solved in $O^*(2^{o(n)})$ time for $n$-vertex bipartite graphs and cannot be solved in $O^*(2^{o(\sqrt{n})})$ time for bipartite graphs with maximum degree~3 and arbitrarily girth.
\item We provide the first exact algorithm to solve \PMC\ on $n$-vertex graphs, of runtime $O^*(1.2721^n)$. Note that the fastest algorithm for \MC\ has runtime $O^*(1.3280^n)$ and is based on the current-fastest algorithm for \textsc{3-sat}~\cite{KomusiewiczKL20}. 
\item We give two graph classes of unbounded mim-width in which \PMC\ is solvable in polynomial time. The first class contains all claw-free graphs and graphs without an induced path on~5 vertices, the second class contains all chordal graphs. 
\end{itemize}

\emph{Related work.} 
The computational complexity of \MC\ was first considered by Chv\'atal 
in~\cite{Chvatal84}, who proved that \MC\ is $\NP$-complete for graphs with maximum degree~$4$ and polynomial time solvable for graphs with maximum degree at most~$3$. Hardness results were obtained for further restricted graph classes such as bipartite graphs, planar graphs and graphs of bounded diameter (see~\cite{Bonsma09,LeL19,Moshi89}). 
Further graph classes in which \MC\ is polynomial time solvable were identified, such as graphs of bounded tree-width, claw-free, hole-free and Ore-graphs (see~\cite{Bonsma09,ChenHLLP21,Moshi89}). 
FPT algorithms and kernelization for \MC\ with respect to various parameters has been discussed in~\cite{AravindKK17,AravindS21,GolovachKKL21,GomesS19,KomusiewiczKL20,KratschL16}. 
The current-best exact algorithm solving \MC\ has a running time of~$O^*(1.3280^n)$ where~$n$ is the vertex number of the input graph~\cite{KomusiewiczKL20}. Faster exact algorithms can be obtained for the case when the minimum degree is large~\cite{ChenHLLP21}. 
The recent paper~\cite{GolovachKKL21} addresses enumeration aspects of matching cuts.   

Very recently, a related notion has been discussed in~\cite{BouquetP20}. In this paper, the authors consider perfect matchings $M\subseteq E$ of a graph $G=(V,E)$ such that $G\setminus M=(V,E\setminus M)$ is disconnected, which they call perfect matching-cuts. To avoid confusion, we call such a perfect matching a \emph{disconnected perfect matching}. Note that, by definition, every perfect matching cut is a disconnected perfect matching but a disconnected perfect matching need not be a perfect matching cut. 
Indeed, all perfect matchings of the cycle on $4k+2$ vertices are disconnected perfect matchings and none of them is a perfect matching cut. 
In~\cite{BouquetP20}, the authors showed, among others, that recognizing graphs having a disconnected perfect matching is $\NP$-complete even when restricted to graphs with maximum degree~$4$, and left open the case of maximum degree~$3$. 
It is not clear whether our hardness result on degree-3 graphs can be modified to obtain a hardness result of recognizing degree-3 graphs having a disconnected perfect matching.    

\emph{Notation and terminology.} 
Let $G=(V,E)$ be a graph with vertex set $V(G)=V$ and edge set $E(G)=E$. 
The neighborhood of a vertex $v$ in $G$, denoted by $N_G(v)$, is the set of all vertices in $G$ adjacent to $v$; if the context is clear, we simply write $N(v)$. 
Let $\deg(v) := |N(v)|$ be the degree of the vertex $v$, and $N[v]:=N(v)\cup\{v\}$ be the closed neighborhood of $v$. 
For a subset~$F\subseteq V$, $G[F]$ is the subgraph of $G$ induced by $F$, and $G-F$ stands for $G[V\setminus F]$. We write $N_F(v)$ and $N_F[v]$ for $N(v) \cap F$ and $N[v]\cap F$, respectively, and call the vertices in $N(v)\cap F$ the \emph{$F$-neighbors} of $v$. 
The \emph{girth} of $G$ is the length of a shortest cycle in $G$, assuming~$G$ contains a cycle. 
The path on $n$ vertices is denoted by $P_n$, the complete bipartite graph with one color class of size $p$ and the other of size $q$ is denoted by $K_{p,q}$; $K_{1,3}$ is also called a \emph{claw}.

When an algorithm branches on the current instance of size $n$ into $r$~subproblems of sizes at most $n-t_1, n-t_2,\ldots, n-t_r$, then $(t_1,t_2,\ldots,t_r)$ is called the \emph{branching vector} of this branching, and the unique positive root of $x^n-x^{n-t_1} - x^{n-t_2} - \cdots - x^{n-t_r} = 0$, denoted by $\tau(t_1,t_2,\ldots,t_r)$, is called its \emph{branching factor}. 
The running time of a branching algorithm is $O^*(\alpha^n)$, where $\alpha=\max_i \alpha_i$ and $\alpha_i$ is the branching factor of branching rule~$i$, and the maximum is taken over all branching rules. Throughout the paper we use the $O^*$ notation which suppresses polynomial factors. 
We refer to~\cite{FominK10} for more details on exact branching algorithms. 

Algorithmic lower bounds in this paper are conditional, based on the Exponential Time Hypothesis (ETH)~\cite{ImpagliazzoP01}. The ETH states that there is no $O^*(2^{o(n)})$-time algorithm for \3SAT\ where $n$ is the variable number of the input \textsc{3-cnf} formula. 
It is known that the hard case for \3SAT\ already consists of formulas with $O(n)$ clauses~\cite{ImpagliazzoPZ01}. 
Thus, assuming ETH, there is no $O^*(2^{o(m)})$-time algorithm for \3SAT\ where $m$ is the clause number of the input formula.  

Observe that a graph has a perfect matching cut if and only if each of its connected components has a perfect matching cut. Thus, we may assume that all graphs in this paper are connected.

\section{Hardness results}\label{sec:hard}
In this section, we give two polynomial time reductions from \NAE3SAT\ to \PMC. Recall that an instance for \NAE3SAT\ is a \textsc{3-cnf} formula $F=C_1\land C_2\land\cdots\land C_m$ over $n$ variables $x_1, x_2, \ldots, x_n$, in which each clause $C_j$ consists of three distinct variables. The problem asks whether there is a truth assignment of the variables such that every clause in $F$ has one true and one false variable. Such an assignment is called \emph{nae assignment}. 

It is well-known that there is a polynomial reduction from \3SAT\ to \NAE3SAT\ where the variable number of the reduced formula is linear in the clause
number of the original formula. Hence, the ETH implies that there is no subexponential time algorithm for \NAE3SAT\ in the number of variables. 
%
\begin{theorem}\label{thm:ETH}
Assuming ETH, \PMC\ cannot be solved in subexponential time in the vertex number, even when restricted to bipartite graphs. 
\end{theorem}
\emph{Proof.}\, 
We give a polynomial reduction from  \NAE3SAT to \PMC\ restricted to bipartite graphs.   
Given a \textsc{3-cnf} formula $F$, construct a graph~$G$ as follows. 

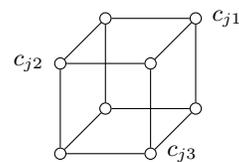
\begin{wrapfigure}{r}{.3\textwidth}
\centering
\begin{tikzpicture}[scale=.3]
\tikzstyle{vertex}=[draw,circle,inner sep=1.5pt]; 
\node[vertex] (cj) at (5,5) {}; 
\node[vertex] (cj1) at (7,7) [label=right:\small $c_{j1}$] {};
\node[vertex] (cj2) at (1,5) [label=left:\small $c_{j2}$] {};
\node[vertex] (cj3) at (5,1) [label=right:\small $c_{j3}$] {};
\node[vertex] (cj') at (3,3) {};
\node[vertex] (cj1') at (1,1) {};
\node[vertex] (cj2') at (7,3) {};
\node[vertex] (cj3') at (3,7) {};

\draw (cj)--(cj1); \draw (cj)--(cj2); \draw (cj)--(cj3);
\draw (cj')--(cj1'); \draw (cj')--(cj2'); \draw (cj')--(cj3');
\draw (cj1)--(cj2'); \draw (cj1)--(cj3');
\draw (cj2)--(cj1'); \draw (cj2)--(cj3');
\draw (cj3)--(cj1'); \draw (cj3)--(cj2');
\end{tikzpicture}
\caption{The graph $G(C_j)$.}\label{fig:cube}
\end{wrapfigure}
For each clause $C_j=\{c_{j1},c_{j2},c_{j3}\}$, let $G(C_j)$ be the cube with \emph{clause vertices} labeled $c_{j1}$, $c_{j2}$, $c_{j3}$, respectively, as depicted in Fig.~\ref{fig:cube}. 
For each variable $x_i$, we introduce a \emph{variable vertex} $x_i$ and a dummy vertex $x_i'$ adjacent only to $x_i$. 
Finally, we connect a variable vertex $x_i$ to a clause vertex in $G(C_j)$ if and only if~$C_j$ contains the variable $x_i$, i.e., $x_i=c_{jk}$ for some $k\in\{1,2,3\}$. 

Observe that $G$ is bipartite and has the following property: no perfect matching $M$ of $G$ (in particular, no perfect matching cut) contains an edge between a clause vertex and a variable vertex. Thus, for every perfect matching cut $M=E(X,Y)$ of~$G$, the restriction $M_j=E(X_j,Y_j)$ on $G(C_j)$ is a perfect matching cut of $G(C_j)$. 
Moreover, $G(C_j)$ has the following property: it has exactly three perfect matching cuts, and in any perfect matching cut of $G(C_j)$ not all clause vertices belong to the same part. Conversely, any bipartition of $C_j$ can be extended (in a unique way) to a perfect matching cut $M_j$ of~$G(C_j)$. See also Fig.~\ref{fig:3cube}. 
\begin{figure}[!ht]
\tikzstyle{vertexX}=[circle,inner sep=1.5pt,fill=black]
\tikzstyle{vertexY}=[draw,circle,inner sep=1.5pt,fill=lightgray]
\tikzstyle{square}=[draw, rectangle,inner sep=1.8pt]
\tikzstyle{vertex}=[draw,circle,inner sep=1.5pt] 
\tikzstyle{subd}=[draw,circle,inner sep=1.5pt]
\begin{center}
\begin{tikzpicture}[scale=.3]
\node[vertexX] (cj) at (5,5) {}; 
\node[vertexY] (cj1) at (7,7) [label=right:\small $c_{j1}$] {};
\node[vertexX] (cj2) at (1,5) [label=left:\small $c_{j2}$] {};
\node[vertexX] (cj3) at (5,1) [label=right:\small $c_{j3}$] {};
\node[vertexY] (cj') at (3,3) {};
\node[vertexX] (cj1') at (1,1) {};
\node[vertexY] (cj2') at (7,3) {};
\node[vertexY] (cj3') at (3,7) {};

\draw[very thick] (cj)--(cj1); \draw (cj)--(cj2); \draw (cj)--(cj3);
\draw[very thick] (cj')--(cj1'); \draw (cj')--(cj2'); \draw (cj')--(cj3');
\draw (cj1)--(cj2'); \draw (cj1)--(cj3');
\draw (cj2)--(cj1'); \draw[very thick] (cj2)--(cj3');
\draw (cj3)--(cj1'); \draw[very thick] (cj3)--(cj2');
\end{tikzpicture}
\qquad
\begin{tikzpicture}[scale=.3] 
\node[vertexX] (cj) at (5,5) {}; 
\node[vertexX] (cj1) at (7,7) [label=right:\small $c_{j1}$] {};
\node[vertexY] (cj2) at (1,5) [label=left:\small $c_{j2}$] {};
\node[vertexX] (cj3) at (5,1) [label=right:\small $c_{j3}$] {};
\node[vertexY] (cj') at (3,3) {};
\node[vertexY] (cj1') at (1,1) {};
\node[vertexX] (cj2') at (7,3) {};
\node[vertexY] (cj3') at (3,7) {};

\draw (cj)--(cj1); \draw[very thick] (cj)--(cj2); \draw (cj)--(cj3);
\draw (cj')--(cj1'); \draw[very thick] (cj')--(cj2'); \draw (cj')--(cj3');
\draw (cj1)--(cj2'); \draw[very thick] (cj1)--(cj3');
\draw (cj2)--(cj1'); \draw (cj2)--(cj3');
\draw[very thick] (cj3)--(cj1'); \draw (cj3)--(cj2');
\end{tikzpicture}
\qquad
\begin{tikzpicture}[scale=.3] 
\node[vertexX] (cj) at (5,5) {}; 
\node[vertexX] (cj1) at (7,7) [label=right:\small $c_{j1}$] {};
\node[vertexX] (cj2) at (1,5) [label=left:\small $c_{j2}$] {};
\node[vertexY] (cj3) at (5,1) [label=right:\small $c_{j3}$] {};
\node[vertexY] (cj') at (3,3) {};
\node[vertexY] (cj1') at (1,1) {};
\node[vertexY] (cj2') at (7,3) {};
\node[vertexX] (cj3') at (3,7) {};

\draw (cj)--(cj1); \draw (cj)--(cj2); \draw[very thick] (cj)--(cj3);
\draw (cj')--(cj1'); \draw (cj')--(cj2'); \draw[very thick] (cj')--(cj3');
\draw[very thick] (cj1)--(cj2'); \draw (cj1)--(cj3');
\draw[very thick] (cj2)--(cj1'); \draw (cj2)--(cj3');
\draw (cj3)--(cj1'); \draw (cj3)--(cj2');
\end{tikzpicture}
\caption{The three perfect matching cuts of $G(C_j)$; black vertices in $X$, gray vertices in~$Y$.}\label{fig:3cube}
\end{center}
\end{figure}
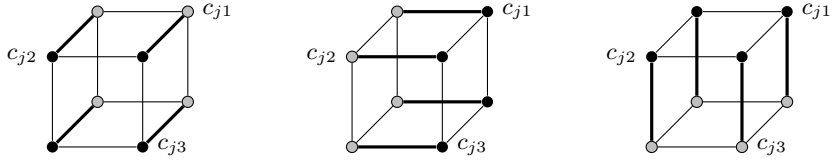

We are now ready to see that $F$ has a nae assignment if and only if $G$ has a perfect matching cut:  
First, if there is a nae assignment for $F$ then put all true variable vertices into $X$, all false variable vertices into $Y$, and extend $X$ and $Y$ (in a unique way) to a perfect matching cut of $G$; note that $x_i'$ and $x_i$ have to belong to different parts. Second, if $(X,Y)$ is a perfect matching cut of $G$ then defining $x_i$ be $\mathtt{true}$ if $x_i\in X$ and $\mathtt{false}$ if $x_i\in Y$ we obtain a nae assignment for~$F$. 

Observe that $G$ has $N=O(n+m)$ vertices. Hence the reduction implies that, assuming ETH, \PMC\ has no subexponential time algorithm in vertex number~$N$, even when restricted to bipartite graphs.\qed
\smallskip

We now describe how to avoid vertices of degree~4 and larger (the clause and variable vertices) in the previous reduction to obtain a bipartite graph with maximum degree~3 and large girth. 

\begin{theorem}\label{thm:bip+deg3+girth}
Let $g>0$ be a given integer. 
\PMC\ remains $\NP$-complete when restricted to bipartite graphs of maximum degree three and girth at least~$g$.
\end{theorem}
\begin{proof} We modify the gadgets used in the proof of Theorem~\ref{thm:ETH}. 
Let $h\ge 0$ be a fixed integer, which will be more concrete later. 

\emph{Clause gadget:} we subdivide every edge of the cube with $4h+4$ new vertices, fix a vertex~$c_j$ of degree~$3$ and label the three neighbors of $c_j$ with $c_{j1}$, $c_{j2}$ and $c_{j3}$, respectively. We denote the obtained graph again by $G(C_j)$ and call the labeled vertices the \emph{clause vertices}. The case $h=0$ is shown in Fig.~\ref{fig:gadgets}. Observe, $G(C_j)$ has the same properties of the cube used in the previous reduction: it has exactly three perfect matching cuts, and in any perfect matching cut of $G(C_j)$ not all clause vertices belong to the same part. Moreover, any bipartition of $C_j$ can be extended (in a unique way) to a perfect matching cut $M_j$ of $G(C_j)$. See also Fig.~\ref{fig:extend}.

\emph{Variable gadget:} for each variable $x_i$ we introduce $m$ \emph{variable vertices} $x_i^j$ one for each clause~$C_j$, $1\le j\le m$, as follows. (We assume that the formula $F$ consists of $m\ge 3$ clauses.) First, take a cycle with $m$ vertices $x_i^1$, $x_i^2$, \ldots, $x_i^m$ and edges $x_i^1x_i^2$, $x_i^2x_i^3$, \ldots, $x_i^{m-1}x_i^m$ and $x_i^1 x_i^m$. 
Then subdivide every edge with $4h+3$ new vertices to obtain the graph $G(x_i)$. Thus, $G(x_i)$ is a cycle on $4m(h+1)$ vertices. The case $m=3, h=0$ is shown in Fig.~\ref{fig:gadgets}. The following property of $G(x_i)$ can be verified immediately: in any perfect matching cut of $G(x_i)$, all variable vertices $x_i^j$, $1\le j\le m$, belong to the same part.  
\begin{figure}[!ht]
\begin{center}
\begin{tikzpicture}[scale=.3]
\tikzstyle{vertexX}=[circle,inner sep=1.5pt,fill=black];
\tikzstyle{vertexY}=[draw,circle,inner sep=1.5pt];
\tikzstyle{square}=[draw, rectangle,inner sep=2pt];
\tikzstyle{vertex}=[draw,circle,inner sep=2pt]; 
\tikzstyle{subd}=[draw,circle,inner sep=1.5pt];

\node[square] (q0) at (11,11) {};
\node[square] (q1) at (16,16) {};
\node[square] (q2) at (1,11) {};
\node[square] (q3) at (11,1) {};
\node[square] (q4) at (6,6) {};
\node[square] (q5) at (16,6) {};
\node[square] (q6) at (6,16) {};
\node[square] (q7) at (1,1) {};

\node[vertex] (cj1) at (12,12) [label=below:\quad $c_{j1}$] {};
\node[vertex] (cj2) at (9,11) [label=below:$c_{j2}$] {};
\node[vertex] (cj3) at (11,9) [label=right:$c_{j3}$] {};

\node[subd] (s01-2) at (13,13) {};
\node[subd] (s01-3) at (14,14) {};
\node[subd] (s01-4) at (15,15) {};
\node[subd] (s02-2) at (7,11) {};
\node[subd] (s02-3) at (5,11) {};
\node[subd] (s02-4) at (3,11) {};
\node[subd] (s03-2) at (11,7) {};
\node[subd] (s03-3) at (11,5) {};
\node[subd] (s03-4) at (11,3) {};

\node[subd] (s15-1) at (16,14) {};
\node[subd] (s15-2) at (16,12) {};
\node[subd] (s15-3) at (16,10) {};
\node[subd] (s15-4) at (16,8) {};
\node[subd] (s16-1) at (14,16) {};
\node[subd] (s16-2) at (12,16) {};
\node[subd] (s16-3) at (10,16) {};
\node[subd] (s16-4) at (8,16) {};

\node[subd] (s26-1) at (2,12) {};
\node[subd] (s26-2) at (3,13) {};
\node[subd] (s26-3) at (4,14) {};
\node[subd] (s26-4) at (5,15) {};
\node[subd] (s27-1) at (1,9) {};
\node[subd] (s27-2) at (1,7) {};
\node[subd] (s27-3) at (1,5) {};
\node[subd] (s27-4) at (1,3) {};

\node[subd] (s35-1) at (12,2) {};
\node[subd] (s35-2) at (13,3) {};
\node[subd] (s35-3) at (14,4) {};
\node[subd] (s35-4) at (15,5) {};
\node[subd] (s37-1) at (9,1) {};
\node[subd] (s37-2) at (7,1) {};
\node[subd] (s37-3) at (5,1) {};
\node[subd] (s37-4) at (3,1) {};

\node[subd] (s45-1) at (8,6) {};
\node[subd] (s45-2) at (10,6) {};
\node[subd] (s45-3) at (12,6) {};
\node[subd] (s45-4) at (14,6) {};
\node[subd] (s46-1) at (6,8) {};
\node[subd] (s46-2) at (6,10) {};
\node[subd] (s46-3) at (6,12) {};
\node[subd] (s46-4) at (6,14) {};
\node[subd] (s47-1) at (5,5) {};
\node[subd] (s47-2) at (4,4) {};
\node[subd] (s47-3) at (3,3) {};
\node[subd] (s47-4) at (2,2) {};

\draw (q0)--(cj1)--(s01-2)--(s01-3)--(s01-4)--(q1);
\draw (q0)--(cj2)--(s02-2)--(s02-3)--(s02-4)--(q2);
\draw (q0)--(cj3)--(s03-2)--(s03-3)--(s03-4)--(q3);

\draw (q1)--(s15-1)--(s15-2)--(s15-3)--(s15-4)--(q5);
\draw (q1)--(s16-1)--(s16-2)--(s16-3)--(s16-4)--(q6);

\draw (q2)--(s26-1)--(s26-2)--(s26-3)--(s26-4)--(q6);
\draw (q2)--(s27-1)--(s27-2)--(s27-3)--(s27-4)--(q7);

\draw (q3)--(s35-1)--(s35-2)--(s35-3)--(s35-4)--(q5);
\draw (q3)--(s37-1)--(s37-2)--(s37-3)--(s37-4)--(q7); 

\draw (q4)--(s45-1)--(s45-2)--(s45-3)--(s45-4)--(q5);
\draw (q4)--(s46-1)--(s46-2)--(s46-3)--(s46-4)--(q6);
\draw (q4)--(s47-1)--(s47-2)--(s47-3)--(s47-4)--(q7);
\end{tikzpicture}
\qquad\qquad
\begin{tikzpicture}[scale=.3] 
\tikzstyle{vertexX}=[circle,inner sep=1.5pt,fill=black];
\tikzstyle{vertexY}=[draw,circle,inner sep=1.5pt];
\tikzstyle{square}=[draw, rectangle,inner sep=1.8pt];
\tikzstyle{vertex}=[draw,circle,inner sep=2pt]; 
\tikzstyle{subd}=[draw,circle,inner sep=1.5pt];

\node[vertex] (vi1) at (13,7) [label=right:$v_i^1$] {};
\node[vertex] (vi2) at (4,12.1966) {}; 
\node at (3.8,13.4) {$v_i^2$};
\node[vertex] (vi3) at (4,1.8034) {}; 
\node at (3.8,0.5) {$v_i^3$};
\node[subd] (s1) at (12.1966,10) {};
\node[subd] (s2) at (10,12.1966) {};
\node[subd] (s3) at (7,13) {};
\node[subd] (s4) at (1.8034,10) {};
\node[subd] (s5) at (1,7) {};
\node[subd] (s6) at (1.8034,4) {};
\node[subd] (s7) at (7,1) {};
\node[subd] (s8) at (10,1.8034) {};
\node[subd] (s9) at (12.1966,4) {};

\centerarc[](7,7)(3:27.5:6); \centerarc[](7,7)(32.5:57.5:6); \centerarc[](7,7)(62.5:87.5:6); \centerarc[](7,7)(92.5:117:6);
\centerarc[](7,7)(123:147.5:6); \centerarc[](7,7)(152.5:177.5:6); \centerarc[](7,7)(182.5:207.5:6); \centerarc[](7,7)(212.5:237:6); 
\centerarc[](7,7)(243:267.5:6); \centerarc[](7,7)(272.5:297.5:6); \centerarc[](7,7)(302.5:327.5:6); \centerarc[](7,7)(332.5:357:6);
\end{tikzpicture}   
\caption{The clause gadget $G(C_j)$ (left) and the variable gadget $G(x_i)$ (right) in case $m=3$ and $h=0$.}\label{fig:gadgets}
\end{center}
\end{figure}
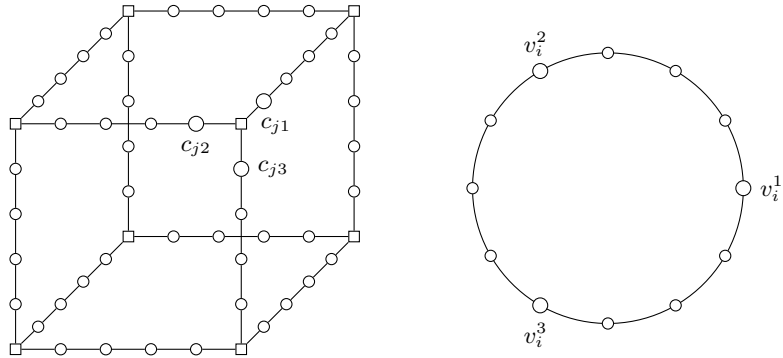

Finally, the graph $G$ is obtained by connecting the variable vertex $x_i^j$ in $G(x_i)$ to a clause vertex in $G(C_j)$ by an edge whenever $x_i$ appears in clause $C_j$, i.e., $x_i=c_{jk}$ for some $k\in\{1,2,3\}$. 
%

It follows from construction, that  
\begin{itemize}
\item $G$ has maximum degree $3$;       
\item $G$ is bipartite. This can be seen as follows. The bipartite subgraph formed by all $G(C_j)$ has a bipartition into independent sets $A$ and $B$ such that all clause  vertices $c_{jk}$ are in~$A$. The bipartite subgraph formed by all $G(x_i)$ has a bipartition into independent sets $C$ and $D$ such that all variable vertices $x_i^j$ are in~$C$. Since the edges in $G$ between these two subgraphs connect clause vertices and variable vertices, therefore the vertex set of $G$ can be partitioned into independent sets $A\cup D$ and $B\cup C$;  
\item $G$ has girth at least $\min\{4m(h+1),8(h+2)\}$. This can be seen as follows. There are~3 types of cycles in $G$. Any of the cycles $G(x_i)$ has length $4m(h+1)$. A shortest cycle in any $G(C_j)$ is a subdivision of a $4$-cycle and has length $4(4h+5)$. The cycles of the last type go through some $G(x_i)$s and some $G(C_j)$s; the length of a shortest one among them is at least $4+(4h+4)+4+(4h+4)=8(h+2)$. 
\end{itemize}
Moreover, as in the previous construction, $G$ has the following property: no perfect matching~$M$ of $G$ (in particular, no perfect matching cut) contains an edge between a clause vertex and a variable vertex. Thus, for every perfect matching cut $M=E(X,Y)$ of $G$, the restrictions of $M$ on $G(C_j)$ and on $G(x_i)$ are perfect matching cuts of $G(C_j)$ and of $G(x_i)$, respectively. 
 
Now, as in the proof of Theorem~\ref{thm:ETH}, we can argue that $F$ has a nae assignment if and only if $G$ has a perfect matching cut. 
First, if there is a nae assignment for~$F$ then put all true variable vertices and clause vertices into $X$, all false variable vertices and clause vertices into $Y$, and extend $X$ and $Y$ (in a unique way) to a perfect matching cut of $G$. See Fig.\ref{fig:extend} for an extension in $G(C_j)$. 
Second, if $(X,Y)$ is a perfect matching cut of $G$ then defining~$x_i$ be $\mathtt{true}$ if $x_i\in X$ and $\mathtt{false}$ if $x_i\in Y$ we obtain a nae assignment for~$F$. 

\begin{figure}[!ht]
\begin{center}
\begin{tikzpicture}[scale=.3]
\tikzstyle{vertexX}=[circle,inner sep=2pt,fill=black];
\tikzstyle{vertexY}=[draw,circle,inner sep=2pt,fill=lightgray];
\tikzstyle{square}=[draw, rectangle,inner sep=2pt];
\tikzstyle{vertex}=[draw,circle,inner sep=2pt]; 
\tikzstyle{subd}=[draw,circle,inner sep=1.5pt];

\node[square] (q0) at (11,11) {};
\node[square] (q1) at (16,16) {};
\node[square] (q2) at (1,11) {};
\node[square] (q3) at (11,1) {};
\node[square] (q4) at (6,6) {};
\node[square] (q5) at (16,6) {};
\node[square] (q6) at (6,16) {};
\node[square] (q7) at (1,1) {};

\node[vertexX] (cj1) at (12,12) [label=below:\quad $c_{j1}$] {};
\node[vertexX] (cj2) at (9,11) [label=below:$c_{j2}$] {};
\node[vertexY] (cj3) at (11,9) [label=right:$c_{j3}$] {};

\node[subd] (s01-2) at (13,13) {};
\node[subd] (s01-3) at (14,14) {};
\node[subd] (s01-4) at (15,15) {};
\node[subd] (s02-2) at (7,11) {};
\node[subd] (s02-3) at (5,11) {};
\node[subd] (s02-4) at (3,11) {};
\node[subd] (s03-2) at (11,7) {};
\node[subd] (s03-3) at (11,5) {};
\node[subd] (s03-4) at (11,3) {};

\node[subd] (s15-1) at (16,14) {};
\node[subd] (s15-2) at (16,12) {};
\node[subd] (s15-3) at (16,10) {};
\node[subd] (s15-4) at (16,8) {};
\node[subd] (s16-1) at (14,16) {};
\node[subd] (s16-2) at (12,16) {};
\node[subd] (s16-3) at (10,16) {};
\node[subd] (s16-4) at (8,16) {};

\node[subd] (s26-1) at (2,12) {};
\node[subd] (s26-2) at (3,13) {};
\node[subd] (s26-3) at (4,14) {};
\node[subd] (s26-4) at (5,15) {};
\node[subd] (s27-1) at (1,9) {};
\node[subd] (s27-2) at (1,7) {};
\node[subd] (s27-3) at (1,5) {};
\node[subd] (s27-4) at (1,3) {};

\node[subd] (s35-1) at (12,2) {};
\node[subd] (s35-2) at (13,3) {};
\node[subd] (s35-3) at (14,4) {};
\node[subd] (s35-4) at (15,5) {};
\node[subd] (s37-1) at (9,1) {};
\node[subd] (s37-2) at (7,1) {};
\node[subd] (s37-3) at (5,1) {};
\node[subd] (s37-4) at (3,1) {};

\node[subd] (s45-1) at (8,6) {};
\node[subd] (s45-2) at (10,6) {};
\node[subd] (s45-3) at (12,6) {};
\node[subd] (s45-4) at (14,6) {};
\node[subd] (s46-1) at (6,8) {};
\node[subd] (s46-2) at (6,10) {};
\node[subd] (s46-3) at (6,12) {};
\node[subd] (s46-4) at (6,14) {};
\node[subd] (s47-1) at (5,5) {};
\node[subd] (s47-2) at (4,4) {};
\node[subd] (s47-3) at (3,3) {};
\node[subd] (s47-4) at (2,2) {};

\draw (q0)--(cj1)--(s01-2)--(s01-3)--(s01-4)--(q1);
\draw (q0)--(cj2)--(s02-2)--(s02-3)--(s02-4)--(q2);
\draw (q0)--(cj3)--(s03-2)--(s03-3)--(s03-4)--(q3);

\draw (q1)--(s15-1)--(s15-2)--(s15-3)--(s15-4)--(q5);
\draw (q1)--(s16-1)--(s16-2)--(s16-3)--(s16-4)--(q6);

\draw (q2)--(s26-1)--(s26-2)--(s26-3)--(s26-4)--(q6);
\draw (q2)--(s27-1)--(s27-2)--(s27-3)--(s27-4)--(q7);

\draw (q3)--(s35-1)--(s35-2)--(s35-3)--(s35-4)--(q5);
\draw (q3)--(s37-1)--(s37-2)--(s37-3)--(s37-4)--(q7); 

\draw (q4)--(s45-1)--(s45-2)--(s45-3)--(s45-4)--(q5);
\draw (q4)--(s46-1)--(s46-2)--(s46-3)--(s46-4)--(q6);
\draw (q4)--(s47-1)--(s47-2)--(s47-3)--(s47-4)--(q7);
\end{tikzpicture}
\qquad\qquad
\begin{tikzpicture}[scale=.3]
\tikzstyle{vertexX}=[circle,inner sep=2pt,fill=black];
\tikzstyle{vertexY}=[draw,circle,inner sep=2pt,fill=lightgray];
\tikzstyle{squareX}=[draw, rectangle,inner sep=2pt,fill=black];
\tikzstyle{squareY}=[draw, rectangle,inner sep=2pt,fill=lightgray];
\tikzstyle{subdX}=[draw,circle,inner sep=1.5pt,fill=black];
\tikzstyle{subdY}=[draw,circle,inner sep=1.5pt,fill=lightgray];

\node[squareX] (q0) at (11,11) {};
\node[squareX] (q1) at (16,16) {};
\node[squareX] (q2) at (1,11) {};
\node[squareY] (q3) at (11,1) {};
\node[squareY] (q4) at (6,6) {};
\node[squareY] (q5) at (16,6) {};
\node[squareX] (q6) at (6,16) {};
\node[squareY] (q7) at (1,1) {};

\node[vertexX] (cj1) at (12,12) [label=below:\quad $c_{j1}$] {};
\node[vertexX] (cj2) at (9,11) [label=below:$c_{j2}$] {};
\node[vertexY] (cj3) at (11,9) [label=right:$c_{j3}$] {};

\node[subdY] (s01-2) at (13,13) {};
\node[subdY] (s01-3) at (14,14) {};
\node[subdX] (s01-4) at (15,15) {};
\node[subdY] (s02-2) at (7,11) {};
\node[subdY] (s02-3) at (5,11) {};
\node[subdX] (s02-4) at (3,11) {};
\node[subdY] (s03-2) at (11,7) {};
\node[subdX] (s03-3) at (11,5) {};
\node[subdX] (s03-4) at (11,3) {};

\node[subdY] (s15-1) at (16,14) {};
\node[subdY] (s15-2) at (16,12) {};
\node[subdX] (s15-3) at (16,10) {};
\node[subdX] (s15-4) at (16,8) {};
\node[subdX] (s16-1) at (14,16) {};
\node[subdY] (s16-2) at (12,16) {};
\node[subdY] (s16-3) at (10,16) {};
\node[subdX] (s16-4) at (8,16) {};

\node[subdX] (s26-1) at (2,12) {};
\node[subdY] (s26-2) at (3,13) {};
\node[subdY] (s26-3) at (4,14) {};
\node[subdX] (s26-4) at (5,15) {};
\node[subdY] (s27-1) at (1,9) {};
\node[subdY] (s27-2) at (1,7) {};
\node[subdX] (s27-3) at (1,5) {};
\node[subdX] (s27-4) at (1,3) {};

\node[subdY] (s35-1) at (12,2) {};
\node[subdX] (s35-2) at (13,3) {};
\node[subdX] (s35-3) at (14,4) {};
\node[subdY] (s35-4) at (15,5) {};
\node[subdY] (s37-1) at (9,1) {};
\node[subdX] (s37-2) at (7,1) {};
\node[subdX] (s37-3) at (5,1) {};
\node[subdY] (s37-4) at (3,1) {};

\node[subdY] (s45-1) at (8,6) {};
\node[subdX] (s45-2) at (10,6) {};
\node[subdX] (s45-3) at (12,6) {};
\node[subdY] (s45-4) at (14,6) {};
\node[subdX] (s46-1) at (6,8) {};
\node[subdX] (s46-2) at (6,10) {};
\node[subdY] (s46-3) at (6,12) {};
\node[subdY] (s46-4) at (6,14) {};
\node[subdY] (s47-1) at (5,5) {};
\node[subdX] (s47-2) at (4,4) {};
\node[subdX] (s47-3) at (3,3) {};
\node[subdY] (s47-4) at (2,2) {};

\draw (q0)--(cj1)--(s01-2)--(s01-3)--(s01-4)--(q1);
\draw (q0)--(cj2)--(s02-2)--(s02-3)--(s02-4)--(q2);
\draw (q0)--(cj3)--(s03-2)--(s03-3)--(s03-4)--(q3);

\draw (q1)--(s15-1)--(s15-2)--(s15-3)--(s15-4)--(q5);
\draw (q1)--(s16-1)--(s16-2)--(s16-3)--(s16-4)--(q6);

\draw (q2)--(s26-1)--(s26-2)--(s26-3)--(s26-4)--(q6);
\draw (q2)--(s27-1)--(s27-2)--(s27-3)--(s27-4)--(q7);

\draw (q3)--(s35-1)--(s35-2)--(s35-3)--(s35-4)--(q5);
\draw (q3)--(s37-1)--(s37-2)--(s37-3)--(s37-4)--(q7); 

\draw (q4)--(s45-1)--(s45-2)--(s45-3)--(s45-4)--(q5);
\draw (q4)--(s46-1)--(s46-2)--(s46-3)--(s46-4)--(q6);
\draw (q4)--(s47-1)--(s47-2)--(s47-3)--(s47-4)--(q7);
\end{tikzpicture}
\caption{How to extend $X$ (black) and $Y$ (gray) on the left-hand site to a perfect matching cut in $G(C_j)$ on the right-hand side.}\label{fig:extend}
\end{center}
\end{figure}
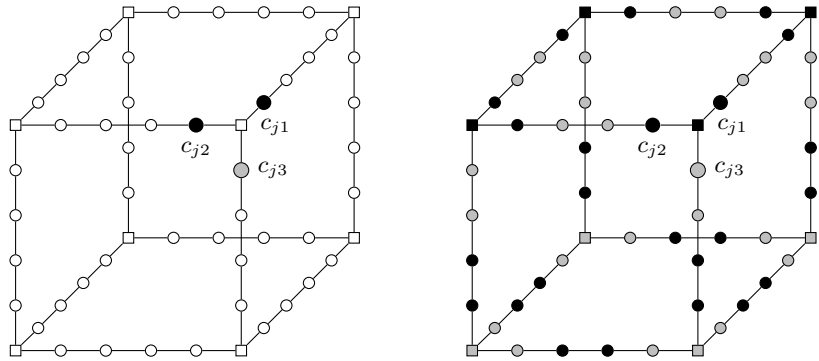

Finally, given $g>0$, let $h\ge 0$ be an integer at least $\max\{\frac{g}{4m}-1,\frac{g}{8}-2\}$. 
Then $G$ has girth at least $\min\{4m(h+1),8(h+2)\}\ge g$. This completes the proof.\qed 
\end{proof}

Note that the graph $G$ in the proof of Theorem~\ref{thm:bip+deg3+girth} has $N=O(m+nm)$ vertices, where $n$ and $m$ are the variable number and clause number, respectively, of the formula $F$. Since we may assume that $F$ has $m=O(n)$ clauses, 
$G$ has $N=O(n^2)$ vertices. Hence we obtain the following.

\begin{theorem}\label{thm:ETH2}
Assuming ETH, there is no $O^*(2^{o(\sqrt{n})})$-time algorithm for \PMC\ even when restricted to $n$-vertex bipartite graphs with maximum degree~$3$ and arbitrary large girth.
\end{theorem}

Observe that \PMC\ is trivial for graphs with maximum degree~2: a (connected) graph with maximum degree~2 has a perfect matching cut if and only if it is a path with even number of vertices or a cycle with $4k$ vertices. Thus, the maximum degree constraint in Theorems~\ref{thm:bip+deg3+girth} and~\ref{thm:ETH2} is optimal.

\section{An exact exponential algorithm}\label{sec:exact}
Recall that, assuming ETH, there is no $O^*(2^{o(n)})$-time algorithm for \PMC\ on $n$-vertex (bipartite) graphs. The main result in this section is an algorithm solving \PMC\ in $O^*(1.2721^n)$ time.

Recall that all graphs considered are connected. 
Our algorithm follows the idea of known branching algorithms for \MC\ \cite{ChenHLLP21,KomusiewiczKL20,KratschL16}. We adapt basic reduction rules for matching cuts to perfect matching cuts, and add new reduction and branching rules for perfect matching cuts.

If the input graph $G=(V,E)$ has a perfect matching cut $(X,Y)$, then some edge has an endvertex $a$ in $X$ and the other endvertex $b$ in $Y$. The  branching algorithm will be executed for all possible edges $ab\in E$, hence $O(m)$ times. 
To do this set $A:=\{a\}$, $B:=\{b\}$, and $F:=V\setminus \{a,b\}$ and call the branching algorithm. At each stage of the algorithm, $A$ and $B$ will be extended or it will be determined that there is no perfect matching cut \emph{separating}~$A$ and~$B$, that is a perfect matching cut $(X,Y)$ with $A\subseteq X$ and $B\subseteq Y$. 
We describe our algorithm by a list of reduction and branching rules given in preference order, i.e., in an execution of the algorithm on any instance of a subproblem one always applies the first rule applicable to the instance, which could be a reduction or a branching rule. A reduction rule produces one subproblem while a branching rule results in at least two subproblems, with different extensions of $A$ and $B$. 
Note that $G$ has a perfect matching cut that separates $A$ from $B$ if and only if in at least one recursive branch, extensions $A'$ of $A$ and $B'$ of $B$ are obtained such that $G$ has a perfect matching cut that separates $A'$ from $B'$.
Typically a rule assigns one or more free vertices, vertices of~$F$, either to $A$ or to $B$ and removes them from~$F$, that is, we always have~$F= V\setminus (A\cup B)$. 

Reduction Rules \ref{rule:R1} (except the last three items), \ref{rule:R2} (except the second item), \ref{rule:R3} and \ref{rule:R4} below are given in~\cite{KratschL16} for matching cuts. As perfect matching cuts are matching cuts, they remain correct for perfect matching cuts.

\begin{rrule}\label{rule:R1} 
\mbox{}
\begin{itemize}
\item
If a vertex in $A$ has two $B$-neighbors, or a vertex in $B$ has two $A$-neighbors then STOP\@: ``$G$ has no matching cut separating $A$, $B$''. 
\item
If $v\in F$, $|N(v)\cap A|\ge 2$ and  $|N(v)\cap B|\ge 2$ then STOP\@: ``$G$ has no matching cut separating $A$, $B$''. 
\item
If there is an edge $xy$ in $G$ such that $x\in A$ and $y\in B$ and $N(x) \cap N(y) \cap F \ne \emptyset$ then STOP\@: ``$G$ has no matching cut separating $A$, $B$''.
\item If a vertex in $A$ and a vertex in $B$ have three or more common neighbors in $F$ then STOP\@: ``$G$ has no matching cut separating $A$, $B$''. 
\item 
If a vertex in $A$ \textup{(}respectively in $B$\textup{)} has no neighbor in $B\cup F$ \textup{(}respectively in $A\cup F$\textup{)} then STOP\@: ``$G$ has no \emph{perfect} matching cut separating $A$, $B$''. 
\item 
If there are $x\in A$ and $y\in B$ such that $N(x)\cap F = N(y)\cap F =\{v\}$ then STOP\@: ``$G$ has no \emph{perfect} matching cut separating $A$, $B$''. 
\end{itemize}
\end{rrule}

\begin{rrule}\label{rule:R2} 
\mbox{}
\begin{itemize}
\item
If $v\in F$ has at least $2$ $A$-neighbors \textup{(}respectively $B$-neighbors\textup{)} then $A:=A\cup \{v\}$ \textup{(}respectively $B:=B\cup\{v\}$\textup{)}. 
\item
If $v\in F$ with $|N(v)\cap N(z)\cap F|\ge 3$ for some $z\in A$ \textup{(}respectively $z\in B$\textup{)} then $A:=A\cup \{v\}\cup (N(v)\cap N(z)\cap F)$ \textup{(}respectively $B:=B\cup \{v\}\cup (N(v)\cap N(z)\cap F)$\textup{)}.  
\end{itemize}
\end{rrule} 

\begin{rrule}\label{rule:R3} 
If $x\in A$ \textup{(}respectively $y\in B$\textup{)} has two adjacent $F$-neighbors $u,v$ then $A:=A\cup \{u,v\}$ \textup{(}respectively $B:=B\cup \{u,v\}$\textup{)}.
\end{rrule}

\begin{rrule}\label{rule:R4} 
If there is an edge $xy$ in $G$ such that $x\in A$ and $y\in B$ then add $N(x) \cap F$ to $A$, and add $N(y)\cap F$ to $B$. 
\end{rrule}

If none of these reduction rules can be applied then the following facts hold: 

\begin{itemize}
\item The edge cut $E(A,B)$ is a (not necessary perfect) matching cut of $G[A\cup B]=G-F$ due to Reduction Rule~\ref{rule:R1}. Moreover, any vertex in $A$ and any vertex in $B$ have at most two common neighbors in $F$.  
\item Every vertex in $F$ is adjacent to at most one vertex in $A$ and at most one vertex in $B$ due to Reduction Rule~\ref{rule:R2}.
\item The neighbors in $F$ of any vertex in $A$ and the neighbors in $F$ of any vertex in $B$ form an independent set due to Reduction Rule~\ref{rule:R3}, and  
\item Every vertex in $A$ adjacent to a vertex in $B$ has no neighbor in $F$ and every vertex in $B$ adjacent to a vertex in $A$ has no neighbor in $F$ due to Reduction Rule~\ref{rule:R4}. 
\end{itemize}

Reduction Rule~\ref{rule:R5} below is given in \cite{KomusiewiczKL20} and remains correct for perfect matching cuts.

\begin{rrule}\label{rule:R5} 
If there are vertices $u,v\in F$ such that $N(u)=N(v)=\{x,y\}$ with $x\in A, y\in B$, then $A:=A\cup\{u\}$, $B:=B\cup\{v\}$.
\end{rrule}

The remaining reduction rules work for perfect matching cuts but not for matching cuts in general.
 
\begin{rrule}\label{rule:R6} 
If $x\in A$ \textup{(}respectively $y\in B$\textup{)} has exactly one neighbor $v\in F$ then $B:=B\cup\{v\}$ \textup{(}respectively $A:=A\cup\{v\}$\textup{)}. 
\end{rrule}
\begin{proof}[of safeness] 
Let $x\in A$ with $N(x)\cap F=\{v\}$. By Reduction Rule~\ref{rule:R4}, $N(x)\cap B=\emptyset$. If $(X,Y)$ is a perfect matching separating $A$ and $B$, then $N(x)\setminus\{v\}\subseteq X$, hence the neighbor~$v$ of~$x$ must belong to~$Y$. The case $y\in B$ is symmetric.
\qed
\end{proof}

\begin{rrule}\label{rule:R7} 
Let $z\in A$ \textup{(}respectively $z\in B$\textup{)} and let $v\in N(z)\cap F$.
\begin{itemize}
\item
If $\deg(v)=1$ then $B:=B\cup\{v\}$ \textup{(}respectively $A:=A\cup\{v\}$\textup{)}. 
\item If $\deg(v)=2$ and $w\in F$ is other neighbor of $v$ then $B:=B\cup\{w\}$ \textup{(}respectively $A:=A\cup\{w\}$\textup{)}.
\end{itemize}
\end{rrule}
\begin{proof}[of safeness] 
Let $z\in A$ and $v\in N(z)\cap F$. Let $(X,Y)$ be a perfect matching of~$G$ separating $A$ and $B$. 
If $z$ is the only neighbor of $v$, then, as $z\in X$, $v$ must belong to~$Y$. 
If $N(v)=\{z,w\}$ with $w\in F$, then $w$ must belong to $Y$, otherwise both neighbors of~$v$ were in $X$. The case $z\in B$ is symmetric.
\qed
\end{proof}

\begin{rrule}\label{rule:R8} 
Let $x\in A$ and $y\in B$ with $|N(x)\cap N(y)\cap F|=2$. 
If $|N(x)\cap F|\ge 3$ or $|N(y)\cap F|\ge 3$ then $A:=A\cup N(x)\setminus N(y)$, $B:=B\cup N(y)\setminus N(x)$.
\end{rrule}
\begin{proof}[of safeness] 
Assume that $(X,Y)$ is a perfect matching cut of~$G$ separating~$A$ and~$B$. Then $N(x)\cap N(y)\cap F$ must contain one vertex in~$X$ and one vertex in~$Y$. Hence $N(x)\setminus N(y)\allowbreak\subseteq X$ and $N(y)\setminus N(x)\subseteq Y$.
\qed
\end{proof}

We now describe the branching rules; see also Fig.~\ref{fig:B1B2B3}, \ref{fig:B4B5} and~\ref{fig:B6B7}. 
The correctness of all branching rules follows from the fact that, in any perfect matching cut $(X,Y)$ separating $A$ and $B$, every vertex in $X$ has exactly one neighbor in $Y$ and every vertex in $Y$ has exactly one neighbor in $X$. 
Thus, if some vertex in $A$ has no neighbor in $B$, it must have a neighbor in~$F$ that must go to~$Y$, and if some vertex in $B$ has no neighbor in~$A$, it must have a neighbor in~$F$ that must go to~$X$. Note that by Reduction Rule~\ref{rule:R6}, every vertex in $A\cup B$ has none or at least two neighbors in~$F$. By Reduction Rule~\ref{rule:R1}, every two vertices $x\in A$ and $y\in B$ have at most two common neighbors in~$F$.

To determine the branching vectors which correspond to our branching rules, we set the size of an instance $(G,A,B)$ as its number of free vertices,
i.e., $|V(G)|-|A|-|B|$. 
Vertices in $A\cup B$ having exactly two neighbors in $F$ will be covered by the first four branching rules. 

\begin{figure}[H]
\tikzstyle{vertexVi}=[circle,inner sep=1.5pt,fill=black]
\tikzstyle{ab}=[circle,inner sep=1.5pt,fill=gray]
\tikzstyle{vertexY}=[draw,circle,inner sep=1.5pt]
\tikzstyle{vertex}=[draw,circle,inner sep=1.5pt] 
\begin{center}
\begin{tikzpicture}[scale=.3]
\filldraw[fill=black!10!white, draw=none] (1,7) rectangle (5,11);
\node[] at (1.8,10.2) {$A$};
\node[vertexVi] (x) at (3.5,8) [label=above:$x$] {};
\filldraw[fill=black!10!white, draw=none] (7,7) rectangle (11,11);
\node[] at (10.2,10.2) {$B$};

\node[ab] (a) at (4.3,10) {};
\node[ab] (b) at (8.2,10) {};
\draw[very thin,gray]  (4,10.5)--(a)--(b)--(9,9.5);
\draw[very thin,gray]  (4.5,10.5)--(a); \draw[very thin,gray] (b)--(8,10.5);
\node[vertexVi] (y) at (8,8)  [label=above:$y$] {}; 

\filldraw[fill=black!10!white, draw=none] (1,1) rectangle (11,5.5);
\node[] at (1.8,1.8) {$F$};
\node[] at (6.5,0) {Branching Rule \ref{rule:B1}};

\node[vertex] (u) at (3,4)  [label=below:$u$] {}; 
\node[vertex] (v) at (8.5,4)  [label=below:$v$] {}; 

\draw[thick] (u) -- (x) -- (v);
\draw[thick] (u) -- (y)--(v);
\end{tikzpicture}
\qquad
\begin{tikzpicture}[scale=.3]
\filldraw[fill=black!10!white, draw=none] (1,7) rectangle (5,11);
\node[] at (1.8,10.2) {$A$};
\node[vertexVi] (x) at (3.5,8) [label=above:$x$] {};
\filldraw[fill=black!10!white, draw=none] (7,7) rectangle (12.5,11);
\node[] at (11.8,10.2) {$B$};

\node[ab] (a) at (4.3,10) {};
\node[ab] (b) at (8.2,10) {};
\draw[very thin,gray]  (4,10.5)--(a)--(b)--(9,9.5);
\draw[very thin,gray]  (4.5,10.5)--(a); \draw[very thin,gray] (b)--(8,10.5);
\node[vertexVi] (y1) at (8,8)  [label=above:$y_1$] {}; 
\node[vertexVi] (y2) at (11,8)  [label=above:$y_2$] {}; 

\filldraw[fill=black!10!white, draw=none] (1,1) rectangle (12.5,5.5);
\node[] at (1.8,1.8) {$F$};
\node[] at (6.5,0) {Branching Rule \ref{rule:B2}};

\node[vertex] (u) at (3,4)  [label=below:$u$] {}; 
\node[vertex] (v) at (6,4)  [label=below:$v$] {}; 
\filldraw[fill=black!18!white, draw] (7.5,3.5) rectangle (9,4.5);
\node[] at (8.3,2.5) {\small $N_1$};
\filldraw[fill=black!18!white, draw] (10,3.5) rectangle (11.5,4.5);
\node[] at (11,2.5) {\small $N_2$};

\draw (7.5,4.4)--(y1)--(9,4.5); 
\draw (10,4.4)--(y2)--(11.5,4.5); 
\draw[thick] (u) -- (x) -- (v);
\draw[thick] (u) -- (y1); \draw[thick] (v)--(y2);
\end{tikzpicture}
\qquad
\begin{tikzpicture}[scale=.3]
\filldraw[fill=black!10!white, draw=none] (1,7) rectangle (5,11);
\node[] at (1.8,10.2) {$A$};
\node[vertexVi] (x) at (3.5,8) [label=above:$x$] {};
\filldraw[fill=black!10!white, draw=none] (7,7) rectangle (11,11);
\node[] at (10.2,10.2) {$B$};

\node[ab] (a) at (4.3,10) {};
\node[ab] (b) at (8.2,10) {};
\draw[very thin,gray]  (4,10.5)--(a)--(b)--(9,9.5);
\draw[very thin,gray]  (4.5,10.5)--(a); \draw[very thin,gray] (b)--(8,10.5);
\node[vertexVi] (y) at (8,8)  [label=above:$y$] {}; 

\filldraw[fill=black!10!white, draw=none] (1,1) rectangle (11,5.5);
\node[] at (1.8,1.8) {$F$};
\node[] at (6.5,0) {Branching Rule \ref{rule:B3}};

\node[vertex] (u) at (3,4)  [label=below:$u$] {}; 
\node[vertex] (v) at (8.5,4)  [label=below:$v$] {}; 

\draw[thick] (u) -- (x) -- (v);
\draw[thick] (y)--(v);
\filldraw[fill=black!18!white, draw] (4,2) rectangle (5.5,3); 
\node at (6.3,2.5) {\small $N$};
\draw (4,2.8)--(u)--(5.5,3);
\end{tikzpicture}
\end{center}
\caption{When Branching Rules \ref{rule:B1}, \ref{rule:B2} and \ref{rule:B3} are applicable.}\label{fig:B1B2B3}
\end{figure}
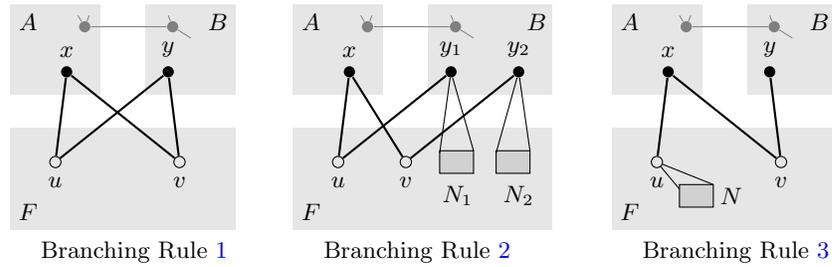

\begin{brule}\label{rule:B1} 
Let $x\in A$ and $y\in B$ with $N(x)\cap N(y)\cap F=\{u,v\}$. By Reduction Rule~\textup{\ref{rule:R8}}, $N(x)\cap F=N(y)\cap F=\{u,v\}$. 
We branch into two subproblems. 
\begin{itemize}
\item
First, add $N[u]\cap F$ to~$A$. Then $N[v]\cap F$ has to be added to~$B$. 
\item
Second, add $N[u]\cap F$ to~$B$. Then $N[v]\cap F$ has to be added to~$A$. 
\end{itemize}
\end{brule}   
The branching vector of Branching Rule~\ref{rule:B1} is 
\[\big(|(N[u]\cup N[v])\cap F|, |(N[u]\cup N[v])\cap F|\big).\] 
By Reduction Rule~\ref{rule:R5}, $|(N[u]\cup N[v])\cap F|\ge 3$, hence the branching factor of Branching Rule~\ref{rule:B1} is at most $\tau(3,3)=\sqrt[3]{2}<1.2560$. 

\begin{brule}\label{rule:B2} 
Let $x\in A$ with $N(x)\cap F=\{u,v\}$ and $N(u)\cap B=\{y_1\}$, $N(v)\cap B=\{y_2\}$. 
We branch into $2$ subproblems. 
\begin{itemize}
\item First, add $u$ to~$B$. Then $v$ has to be added to $A$ and $N_2:= N(y_2)\cap F\setminus\{v\}$ has to be added to~$B$.
\item Second, add $v$ to~$B$. Then $u$ has to be added to $A$ and $N_1:= N(y_1)\cap F\setminus\{u\}$ has to be added to~$B$.
\end{itemize}
Symmetrically for $y\in B$ with $N(y)\cap F=\{u,v\}$ and $N(u)\cap A=\{x_1\}$, $N(v)\cap A=\{x_2\}$.  
\end{brule}
By Branching Rule~\ref{rule:B1}, $v\not\in N_1$, $u\not\in N_2$. Hence, the branching vector of Branching Rule~\ref{rule:B2} is 
\[\big(2+|N_1|, 2+|N_2|\big).\] 
By Reduction Rule~\ref{rule:R6}, $|N_1|\ge 1, |N_2|\ge 1$. 
Hence the branching factor is at most $\tau(3,3)=\sqrt[3]{2}<1.2560$. 

\begin{brule}\label{rule:B3} 
Let $x\in A$ with $N(x)\cap F=\{u,v\}$ and $N(u)\cap B=\emptyset$, $N(v)\cap B=\{y\}$.
We branch into two subproblems. 
\begin{itemize}
\item
First, add $u$ to~$B$. Then $v$ has to be added to $A$ and $N:=N(u)\cap F$ has to be added to~$B$. 
\item
Second, add $v$ to~$B$. Then $u$ has to be added to~$A$. 
\end{itemize}
Symmetrically for $y\in B$ with $N(y)\cap F=\{u,v\}$ and $N(u)\cap A=\emptyset$, and $N(v)\cap A=\{x\}$.
\end{brule}
The branching vector of Branching Rule~\ref{rule:B3} is 
\[\big(2+|N|, 2\big).\] 
By Reduction Rule~\ref{rule:R7}, $|N|\ge 2$, hence the branching factor of Branching Rule~\ref{rule:B3} is at most $\tau(4,2)<\textcolor{red}{1.2721}$. 

\begin{brule}\label{rule:B4} 
Let $x\in A$ with $N(x)\cap F=\{u_1,u_2,\ldots,u_r\}$, $r\ge 2$, and $N(u_i)\cap B\allowbreak =\emptyset$, $1\le i\le r$. 
We branch into $r$ subproblems. For each $1\le i\le r$, the instance of the $i$-th subproblem is obtained by adding $u_i$ to~$B$. Then $N(x)\cap F\setminus\{u_i\}$ has to be added to~$A$ and $N_i:=N(u_i)\cap F$ has to be added to~$B$.\\
Symmetrically for $y\in B$ with $N(y)\cap F=\{v_1,v_2,\ldots,v_r\}$ and $v_i$ has no neighbor in $A$, $1\le i\le r$.  
\end{brule}  
The branching vector of Branching Rule~\ref{rule:B4} is 
\[\big(r+|N_1|, r+|N_2|,\ldots, r+|N_r|\big).\]  
By Reduction Rule~\ref{rule:R7}, $|N_i|\ge 2$, hence the branching factor of Branching Rule~\ref{rule:B4} is at most $\tau(r+2,r+2,\ldots, r+2)=\sqrt[r+2]{r}<1.2600$. 

\begin{figure}[H]
\tikzstyle{vertexVi}=[circle,inner sep=1.5pt,fill=black]
\tikzstyle{ab}=[circle,inner sep=1.5pt,fill=gray]
\tikzstyle{vertexY}=[draw,circle,inner sep=1.5pt]
\tikzstyle{vertex}=[draw,circle,inner sep=1.5pt] 
\begin{center}
\begin{tikzpicture}[scale=.3]
\filldraw[fill=black!10!white, draw=none] (1,9) rectangle (8,13);
\node[] at (1.8,12.2) {$A$};
\node[vertexVi] (x) at (5,10) [label=above:$x$] {};
\filldraw[fill=black!10!white, draw=none] (10,9) rectangle (18,13);
\node[] at (17.2,12.2) {$B$};

\node[ab] (a) at (7,12) {};
\node[ab] (b) at (11,12) {};
\draw[very thin,gray]  (6,12.5)--(a)--(b)--(12,12.5);
\draw[very thin,gray]  (6.5,11)--(a); \draw[very thin,gray] (b)--(12.5,11.5);

\filldraw[fill=black!10!white, draw=none] (1,1) rectangle (18,7);
\node[] at (1.8,1.8) {$F$};
\node[] at (8.5,0) {Branching Rule \ref{rule:B4}};

\node[vertex] (u1) at (6,4) [label=below:$u_1$] {}; 
\node[vertex] (ui) at (9,4) [label=below:$u_i$] {}; 
\node[vertex] (ur) at (12.5,5) [label=below:$u_r$] {}; 
\filldraw[fill=black!18!white, draw] (11,2) rectangle (13,3);
\node[] at (13.8,2.5) {\small $N_i$};
\draw (11,2.3)--(ui)--(13,3);

\draw[thick] (u1) -- (x) -- (ui);
\draw[thick] (x) -- (ur);

\end{tikzpicture}
\qquad
\begin{tikzpicture}[scale=.3]
\filldraw[fill=black!10!white, draw=none] (1,9) rectangle (8,13);
\node[] at (1.8,12.2) {$A$};
\node[vertexVi] (x) at (5,10) [label=above:$x$] {};
\filldraw[fill=black!10!white, draw=none] (10,9) rectangle (18,13);
\node[] at (17.2,12.2) {$B$};

\node[ab] (a) at (7,12) {};
\node[ab] (b) at (11,12) {};
\draw[very thin,gray]  (6,12.5)--(a)--(b)--(12,12.5);
\draw[very thin,gray]  (6.5,11)--(a); \draw[very thin,gray] (b)--(12.5,11.7);
\node[vertexVi] (y1) at (12,10) [label=above:$y_1$] {}; 
\node[vertexVi] (yj) at (14,10) [label=above:$y_j$] {}; 
\node[vertexVi] (yq) at (16,10) [label=above:$y_q$] {}; 

\filldraw[fill=black!10!white, draw=none] (1,1) rectangle (18,7);
\node[] at (1.8,1.8) {$F$};
\node[] at (8.5,0) {Branching Rule \ref{rule:B5}};

\node[vertex] (u1) at (3,4)  [label=below:$u_1$] {}; 
\node[vertex] (ui) at (5,4)  [label=below:$u_i$] {}; 
\node[vertex] (up) at (7,4)  [label=below:$u_p$] {}; 
\node[vertex] (v1) at (9,4)  [label=below:$v_1$] {}; 
\node[vertex] (vj) at (11,4) [label=below:$v_j$] {}; 
\node[vertex] (vq) at (13,4) [label=below:$v_q$] {}; 
\filldraw[fill=black!18!white, draw] (15,3.5) rectangle (17,4.5);
\node[] at (16,2.5) {\small $N_j$};

\draw (15,4.3)--(yj)--(17,4.5); 
\draw (12.2,8.5)--(y1)--(12.6,8.6); \draw (16.2,8.5)--(yq)--(16.6,8.6);
\draw[thick] (u1) -- (x) -- (ui);
\draw[thick] (up) -- (x) -- (v1);
\draw[thick] (vj) -- (x) -- (vq);
\draw[thick] (v1) -- (y1); \draw[thick] (vj) -- (yj); \draw[thick] (vq) -- (yq);
\end{tikzpicture}
\end{center}
\caption{When Branching Rules \ref{rule:B4} and \ref{rule:B5} are applicable.}\label{fig:B4B5}
\end{figure}
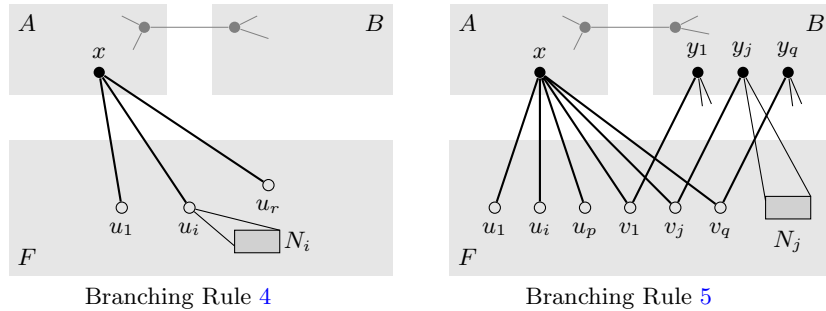

Branching Rules~\ref{rule:B1} and~\ref{rule:B4} together with the remaining branching rules cover vertices in $A\cup B$ having at least three neighbors in $F$. 
Branching Rule~\ref{rule:B5} deals with the case $z\in A$ (respectively $z\in B$) in which at least two vertices in $N(z)\cap F$ have neighbors in $B$ (respectively in $A$).

\begin{brule}\label{rule:B5} 
Let $x\in A$ with $N(x)\cap F=\{u_1,\ldots, u_p, v_1, v_2,\ldots,v_q\}$, $p\ge 0$, $q\ge 2$, such that $N(u_i)\cap B=\emptyset$, $1\le i\le p$ and $N(v_j)\cap B=\{y_j\}$, $1\le j\le q$.  
We branch into $r=p+q$ subproblems. 
\begin{itemize}
\item
For each $1\le i\le p$, the instance of the $i$-th subproblem is obtained by adding $u_i$ to~$B$. Then $N(x)\cap F\setminus\{u_i\}$ has to be added to~$A$ and all $N_j:= N(y_j)\cap F\setminus\{v_j\}$, $1\le j\le q$, have to be added to~$B$. 
\item
For each $1\le j\le q$, the instance of the $p+j$-th subproblem is obtained by adding $v_j$ to~$B$. Then $N(x)\cap F\setminus\{v_j\}$ has to be added to $A$ and all $N_k:= N(y_j)\cap F\setminus\{v_j\}$, $1\le k\le q$, $k\not=j$, have to be added to $B$.
\end{itemize}
Symmetrically for $y\in B$ with $N(y)\cap F=\{u_1,\ldots, u_p, v_1, v_2,\ldots,v_q\}$, $p\ge0$, $q\ge2$ such that $N(u_i)\cap A=\emptyset$, $1\le i\le p$ and $N(v_j)\cap A=\{x_j\}$, $1\le j\le q$. 
\end{brule}
By Branching Rule~\ref{rule:B1} and Reduction Rule~\ref{rule:R2}, $N_j$ are pairwise disjoint and $N_j\cap\{v_1,\ldots,v_q\}\allowbreak =\emptyset$. 
Hence, the branching vector of Branching Rule~\ref{rule:B5} is 
\[\big(r+\sum_j|N_j|, \dots, r+\sum_j|N_j|, r+\sum_{k\not=1}|N_k|, \dots, r+\sum_{k\not=q}|N_k|\big).\] 
Due to Branching Rules~\ref{rule:B1}--\ref{rule:B4}, each $y_j$ has at least three neighbors in $F$.  
Hence $|N_j|\ge 2$, $1\le j\le q$. Thus, the branching factor is at most $\tau(r+2q,\dots, r+2q, r+2(q-1),\dots, r+2(q-1))\le  \tau(r+4,\ldots,r+4,r+2,\dots,r+2)<\tau(r+2,\dots,r+2)=\sqrt[r+2]{r}<1.2600.$ 

The last two branching rules deal with the case $z\in A$ (respectively $z\in B$) in which exactly one vertex in $N(z)\cap F$ has a unique neighbor in $B$ (respectively in $A$).

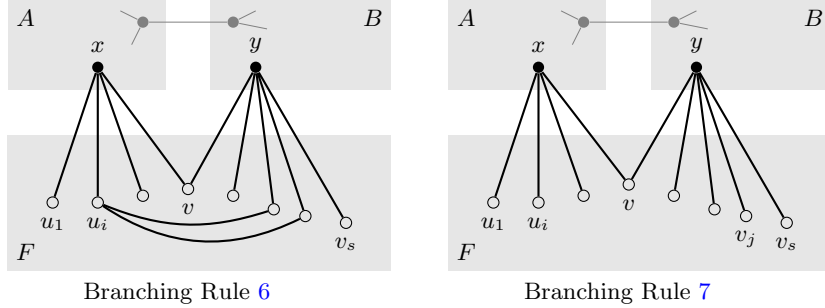
\begin{figure}[H]
\tikzstyle{vertexVi}=[circle,inner sep=1.5pt,fill=black]
\tikzstyle{ab}=[circle,inner sep=1.5pt,fill=gray]
\tikzstyle{vertexY}=[draw,circle,inner sep=1.5pt]
\tikzstyle{vertex}=[draw,circle,inner sep=1.5pt] 
\begin{center}
\begin{tikzpicture}[scale=.3]
\filldraw[fill=black!10!white, draw=none] (1,9) rectangle (8,13);
\node[] at (1.8,12.2) {$A$};
\node[vertexVi] (x) at (5,10) [label=above:$x$] {};
\filldraw[fill=black!10!white, draw=none] (10,9) rectangle (18,13);
\node[] at (17.2,12.2) {$B$};

\node[ab] (a) at (7,12) {};
\node[ab] (b) at (11,12) {};
\draw[very thin,gray]  (6,12.5)--(a)--(b)--(12,12.5);
\draw[very thin,gray]  (6.5,11)--(a); \draw[very thin,gray] (b)--(12.5,11.7);
\node[vertexVi] (y) at (12,10)  [label=above:$y$] {}; 

\filldraw[fill=black!10!white, draw=none] (1,1) rectangle (18,7);
\node[] at (1.8,1.8) {$F$};
\node[] at (8.5,0) {Branching Rule \ref{rule:B6}};

\node[vertex] (u1) at (3,4) [label=below:$u_1$] {}; 
\node[vertex] (ui) at (5,4)  [label=below:$u_i$] {}; 
\node[vertex] (ur) at (7,4.3)  {}; 
\node[vertex] (v) at (9,4.6)  [label=below:$v$] {}; 
\node[vertex] (v1) at (11,4.3) {}; 
\node[vertex] (vj) at (12.8,3.7) {};
\node[vertex] (vk) at (14.2,3.4) {};
\node[vertex] (vs) at (16,3.1) [label=below:$v_s$] {}; 

\draw[thick] (u1) -- (x) -- (ui);
\draw[thick] (ur) -- (x) -- (v);
\draw[thick] (v) -- (y); \draw[thick] (v1) -- (y) --(vj); \draw[thick] (vk) -- (y)--(vs);
\draw[thick] (ui) to[bend angle=20, bend right] (vj); \draw[thick] (ui) to[bend angle=30, bend right] (vk); 
\end{tikzpicture}
\qquad
\begin{tikzpicture}[scale=.3]
\filldraw[fill=black!10!white, draw=none] (1,9) rectangle (8,13);
\node[] at (1.8,12.2) {$A$};
\node[vertexVi] (x) at (5,10) [label=above:$x$] {};
\filldraw[fill=black!10!white, draw=none] (10,9) rectangle (18,13);
\node[] at (17.2,12.2) {$B$};

\node[ab] (a) at (7,12) {};
\node[ab] (b) at (11,12) {};
\draw[very thin,gray]  (6,12.5)--(a)--(b)--(12,12.5);
\draw[very thin,gray]  (6.5,11)--(a); \draw[very thin,gray] (b)--(12.5,11.7);
\node[vertexVi] (y) at (12,10)  [label=above:$y$] {}; 

\filldraw[fill=black!10!white, draw=none] (1,1) rectangle (18,7);
\node[] at (1.8,1.8) {$F$};
\node[] at (8.5,0) {Branching Rule \ref{rule:B7}};

\node[vertex] (u1) at (3,4) [label=below:$u_1$] {}; 
\node[vertex] (ui) at (5,4)  [label=below:$u_i$] {}; 
\node[vertex] (ur) at (7,4.3)  {}; 
\node[vertex] (v) at (9,4.8)  [label=below:$v$] {}; 
\node[vertex] (v1) at (11,4.3) {}; 
\node[vertex] (vj) at (12.8,3.7) {};
\node[vertex] (vk) at (14.2,3.4) [label=below:$v_j$] {}; 
\node[vertex] (vs) at (16,3.1)  [label=below:$v_s$] {}; 

\draw[thick] (u1) -- (x) -- (ui);
\draw[thick] (ur) -- (x) -- (v);
\draw[thick] (v) -- (y); \draw[thick] (v1) -- (y) --(vj); \draw[thick] (vk) -- (y)--(vs);
\end{tikzpicture}
\end{center}
\caption{When Branching Rules \ref{rule:B6} and \ref{rule:B7} are applicable.}\label{fig:B6B7}
\end{figure}

\begin{brule}\label{rule:B6} 
Let $x\in A$ with $N(x)\cap F=\{u_1, u_2,\ldots, u_r, v\}$, $r\ge 2$, such that $N(u_i)\cap B=\emptyset$, $1\le i\le r$, and $N(v)\cap B=\{y\}$.   
Write $N(y)\cap F\setminus\{v\}=\{v_1,\ldots,v_s\}$, $s\ge 2$. 
Assume that some $u_i$ has two neighbors in $\{v_1,\ldots,v_s\}$.    
We branch into $2$ subproblems. 
\begin{itemize}
\item First, add $v$ to $A$. Then $\{v_1,\ldots,v_s\}$ and $u_i$ have to be added to $B$, and $\{u_1,\ldots,u_r\}\setminus\{u_i\}$ has to be added to $A$. 
\item Second, add $v$ to $B$. Then $\{u_1,\ldots,u_r\}$ has to be added to $A$.
\end{itemize}
Symmetrically for $y\in B$ with $N(y)\cap F=\{u_1, u_2,\ldots, u_{r}, v\}$ such that $N(u_i)\cap A=\emptyset$, $1\le i\le r$, and $N(v)\cap A=\{x\}$ and some $u_i$ has two neighbors in $N(x)\cap F\setminus\{v\}$. 
\end{brule}
The branching vector of Branching Rule~\ref{rule:B6} is 
\[(r+s+1,r+1).\]
Since $r\ge 2$ and $s\ge 2$, we have $\tau(r+s+1,r+1)\le \tau(5,3)< 1.1939$.

\begin{brule}\label{rule:B7} 
Let $x\in A$ with $N(x)\cap F=\{u_1, u_2,\ldots, u_r, v\}$, $r\ge 2$, such that $N(u_i)\cap B=\emptyset$, $1\le i\le r$, and $N(v)\cap B=\{y\}$.   
Write $N(y)\cap F\setminus\{v\}=\{v_1,\ldots,v_s\}$, $s\ge 2$. 
We branch into $r+s$ subproblems. 
\begin{itemize}
\item
For each $1\le i\le r$, the instance of the $i$-th subproblem is obtained by adding $u_i$ to $B$. 
Then $\{u_1,\ldots,u_r\}\setminus\{u_i\}$ and $v$ have to be added to $A$, $N_i:=N(u_i)\cap F$ and $\{v_1,\dots,v_s\}$ have to be added to $B$. 
\item
For each $1\le j\le s$, the instance of the $r+j$-th subproblem is obtained by adding $v_j$ to $A$. 
Then $\{v_1,\dots,v_s\}\setminus\{v_j\}$ and $v$ have to be added to $B$, $M_j:=N(v_j)\cap F$ and $\{u_1,\dots,u_{r}\}$ have to be added to $A$. 
\end{itemize}
Symmetrically for $y\in B$ with $N(y)\cap F=\{u_1, u_2,\ldots, u_{r}, v\}$ such that $N(u_i)\cap A=\emptyset$, $1\le i\le r$, and $N(v)\cap A=\{x\}$. 
\end{brule}  
Write $\alpha_i=|N_i\cap\{v_1,\dots,v_s\}|$, $1\le i\le r$, and $\beta_j=|M_j\cap\{u_1,\dots,u_{r}\}|$, $1\le j\le s$. The branching vector of Branching Rule~\ref{rule:B7} is 
\[\big(r+s+1+|N_1|-\alpha_1,\dots, r+s+1+|N_{r}|-\alpha_{r}, r+s+1+|M_1|-\beta_1,\dots,r+s+1+|M_s|-\beta_s\big).\] 
By Reduction Rule~\ref{rule:R7}, $|N_i|\ge 2$. By Branching Rule~\ref{rule:B5}, $v_j$ has no neighbor in $A$, hence, by Reduction Rule~\ref{rule:R7}, $|M_j|\ge 2$. 
By Branching Rule~\ref{rule:B6}, $\alpha_i\le 1$, $\beta_j\le 1$. 
Hence the branching factor is at most $\tau(r+s+2,\dots,r+s+2)=\sqrt[r+s+2]{r+s}<1.2600$. 

The description of all branching rules is completed. 
Among all branching rules, Branching Rule~\ref{rule:B3} has the largest branching factor of~$1.2721$. Consequently, the running time of our algorithm is~$O^*(1.2721^{n})$. 

It remains to show that if none of the reduction rules and none of the branching rules is applicable to an instance $(G,A,B)$ then the graph $G$ has a perfect matching cut $(X,Y)$ such that $A\subseteq X$ and $B\subseteq Y$ if and only if $(A,B)$ is a perfect matching cut of $G$. 
In fact, if all reduction and branching rules are not longer applicable, then no vertex in $A\cup B$ has a neighbor in $F$. 
Hence, by connectedness of~$G$, $F=\emptyset$. Therefore, $G$ has a perfect matching cut separating $A$ and $B$ if and only if $(A,B)$ is a perfect matching cut. In summary, we obtain:

\begin{theorem}\label{thm:exact}
There is an algorithm for \PMC\ running in $O^*(1.2721^n)$ time.
\end{theorem}  

\section{Two polynomial solvable cases}\label{sec:polytime}
In this section, we provide two graph classes in which \PMC\ is solvable in polynomial time. Both classes are well motivated by the hardness results. 

\subsection{Excluding a \textup{(}small\textup{)} tree of maximum degree three}
Let $H$ be a fixed graph. A graph $G$ is $H$-free if $G$ contains no induced subgraph isomorphic to $H$. Since by Theorem~\ref{thm:bip+deg3+girth} 
\PMC\ remains $\NP$-complete on the class of graphs of maximum degree three and arbitrarily high girth, it is also $\NP$-complete on $H$-free graphs whenever~$H$ is outside this class, e.g. if $H$ has a vertex of degree larger than three or has a (fixed-size) cycle.
This suggests studying the computational complexity of \PMC\ restricted to $H$-free graphs for a fixed forest~$H$ with maximum degree at most three.  

\begin{wrapfigure}{r}{.25\textwidth}
\centering
\begin{tikzpicture}[scale=.3] 
\tikzstyle{vertex}=[draw,circle,inner sep=1.5pt] 
\node[vertex] (a) at (1,1) {};
\node[vertex] (b) at (3,1) {};
\node[vertex] (c) at (5,1) {};
\node[vertex] (d) at (7,1) {};
\node[vertex] (e) at (9,1) {};
\node[vertex] (f) at (5,3) {};

\draw (a)--(b)--(c)--(d)--(e);
\draw (c)--(f);
\end{tikzpicture}
\caption{The tree $T$.}\label{fig:T} 
\end{wrapfigure}
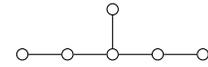
As the first step in this direction, we show that \PMC\ is solvable in polynomial time for $H$-free graphs, where $H$ is the tree $T$ with~6 vertices  obtained from the claw $K_{1,3}$ by subdividing two edges each with one new vertex; see Fig.~\ref{fig:T}. In particular, \PMC\ is polynomial time solvable for $K_{1,3}$-free graphs but hard for $K_{1,4}$-free graphs (by Theorem~\ref{thm:bip+deg3+girth}).

Given a connected $T$-free graph $G=(V,E)$, our algorithm works as follows. Fix an edge $ab\in E$ and decide if $G$ has a perfect matching cut $M=E(X,Y)$ separating~$A=\{a\}$ and~$B=\{b\}$. We use the notations and reduction rules from Section~\ref{sec:exact}. In addition, we need one new reduction rule; recall that $F=V\setminus(A\cup B)$. This additional reduction rule is correct for matching cuts in general and is already used in~\cite{ChenHLLP21}. For completeness, we give a correctness proof for perfect matching cuts. 
\begin{rrule}\label{rule:R10}
\mbox{}
\begin{itemize}
\item If there are vertices $u,v\in F$ with a common neighbor in $A$ and $|N(u)\cap N(v)\cap F|\ge 2$, then $A:=A\cup\{u,v\}$. 
\item If there are vertices $u,v\in F$ with a common neighbor in $B$ and $|N(u)\cap N(v)\cap F|\ge 2$, then $B:=B\cup\{u,v\}$.
\end{itemize}
\end{rrule}
\begin{proof}[of safeness] 
Let $u,v\in F$ with $N(u)\cap N(v)\cap A=\{x\}$ and $|N(u)\cap N(v)\cap F|\ge 2$.  
We show that $G$ has a perfect matching cut separating $A$, $B$ if and only if $G$ has a perfect matching cut separating $A\cup\{u,v\}$ and $B$. 
First, let $(X,Y)$ be a perfect matching cut of $G$ with $A\subseteq X$ and $B\subseteq Y$. If $u\in Y$ then, as $x\in A$, $N(u)\cap F$ must belong to $Y$ and $v$ must belong to $X$. But then, as $|N(v)\cap N(u)\cap F|\ge 2$, $v$ has two neighbors in $Y$, a contradiction. Thus, $u\in X$, and similarly, $v\in X$. That is $(X,Y)$ separates $A\cup\{u,v\}$ and $B$.  
The other direction is obvious: any perfect matching cut separating $A\cup\{u,v\}$ and~$B$ separates~$A$ and~$B$.

The second case is symmetric.\qed
\end{proof}

Now, we apply the Reduction Rules~\ref{rule:R1}--\ref{rule:R10} exhaustively. Note that this part takes polynomial time. If $F=V\setminus (A\cup B)$ is empty, then $G$ has a perfect matching cut separating $A$ and $B$ if and only if $(A,B)$ is a perfect matching cut of $G$. Verifying whether $(A,B)$ is a perfect matching cut also takes polynomial time. 

So, let us assume that $F\not=\emptyset$. Then due to the reduction rules (recall that $G$ is connected),   
\begin{itemize}
\item any vertex in $A$ (in $B$) having no neighbor in $B$ (in $A$) has at least two neighbors in $F$, 
and
\item any vertex in $A$ (in $B$) having a neighbor in $F$ has no neighbor in $B$ (in $A$).
\end{itemize}    
At this point, we will explicitly give an induced subgraph in $G$ isomorphic to the tree $T$ or correctly decide that $G$ has no perfect matching cut separating $A$ and $B$. Write 
\[A^*=\{x\in A\mid N(x)\cap B\not=\emptyset\},\, B^*=\{y\in B\mid N(y)\cap A\not=\emptyset\}.\]
Recall that $A^*\not=\emptyset$ and $B^*\not=\emptyset$, and there are no edges between $A^*\cup B^*$ and $F$, no edges between $A\setminus A^*$ and $B\setminus B^*$. 

Thus, as $G$ is connected and $F\not=\emptyset$, there is a vertex in $A\setminus A^*$ adjacent to a vertex in $A^*$, or there is a vertex in $B\setminus B^*$ adjacent to a vertex in $B^*$. 
By symmetry, let us assume that there is a vertex $x\in A\setminus A^*$ adjacent to a vertex $x^*\in A^*$. Let $y^*\in B^*$ be the unique neighbor of~$x$ in $B^*$. 
Recall that, every vertex in $(A\setminus A^*)\cup (B\setminus B^*)$ has at least two neighbors in $F$. 

First, suppose that there is a vertex $y\in B$ with $|N(x)\cap N(y)\cap F|\ge 2$. 
Let $u,v\in N(x)\cap N(y)\cap F$. If $(X,Y)$ is a perfect matching cut with $A\subseteq X$ and $B\subseteq Y$, then $u$ and $v$ must belong to different parts, say $u\in X, v\in Y$. Now, if there were some vertex $w\in N(u)\cap N(v)\cap F$, then $u$ would have two neighbors in $Y$ (if $w\in Y$) or $v$ would have two neighbors in $X$ (if $w\in X$). 
So, let us assume that $N(u)\cap N(v)\cap F=\emptyset$. Then due to Reduction Rule~\ref{rule:R5}, there exists a vertex $w\in N(u)\cap F\setminus N(v)$. Due to Reduction Rule~\ref{rule:R3}, $N(x)\cap F$ is an independent set, hence $w,u,x,x^*, y^*$ and $v$ induce the tree~$T$ in~$G$. 
Thus, we may assume that 
\begin{align}\label{fact1}
&\text{for any vertex $y\in B$, $|N(x)\cap N(y)|\le 1$.}
\end{align} 

Next, observe that  
\begin{align}\label{fact2}
&\text{every vertex in $B\setminus B^*$ adjacent to a vertex in $N(x)$ is adjacent to $y^*$.}
\end{align}
This can be seen as follows: Let $z\in B\setminus B^*$ be adjacent to some $u\in N(x)$. Then $u\in F$.  By~(\ref{fact1}),~$z$ is non-adjacent to all vertices in $N(x)\cap F\setminus\{u\}$. Recall that some vertex $v\in N(x)\cap F\setminus\{u\}$ exists. So, if $z$ is not adjacent to $y^*$, then $z,u,x,x^*,y^*$ and $v$ induce the tree~$T$ in~$G$. 

Now, fix two vertices $u, v\in N(x)\cap F$. 
Suppose that $N(u)\cap B=\emptyset$. Then, due to Reduction Rules~\ref{rule:R7} and~\ref{rule:R10}, there exists a vertex $w\in N(u)\cap F\setminus N(v)$, and as above, $w,u,v,x,x^*$ and~$y^*$ induce the tree~$T$ in~$G$. 
Thus, we may assume that $N(u)\cap B\not=\emptyset$ and, by symmetry, $N(v)\cap B\not=\emptyset$. 

Let $y_1,y_2\in B\setminus B^*$ be the unique neighbors of $u$ and $v$ in $B$, respectively. 
By~(\ref{fact1}), $y_1$ is non-adjacent to $v$, and $y_2$ is non-adjacent to $u$. By~(\ref{fact2}), $y_1$ and $y_2$ are adjacent to $y^*$. 
If $y_1$ and $y_2$ are non-adjacent, then $u, y_1, y^*, y_2, v$ and $x^*$ induce the tree $T$. So, let us assume that~$y_1$ and~$y_2$ are adjacent. 
 
Let $u'\not=u$ be a second neighbor of $y_1$ in $F$, and $v'\not=v$ be a second neighbor of $y_2$ in $F$. 
By~(\ref{fact1}), $x$ is non-adjacent to $u'$ and $v'$. 
Now, consider two cases: 
\begin{itemize}
\item assume that $u$ and $v'$ are adjacent. Then $v', u, x, x^*, y^*$ and $v$ induce the tree~$T$ in~$G$, and 
\item assume that $u$ and $v'$ are non-adjacent. Then $v', y_2, y_1, u, x$ and $u'$ (if $u'$ and $v'$ are non-adjacent), or else $u', v', y_2, y^*, x^*$ and $v$ (if $u'$ and $v'$ are adjacent) induce the tree~$T$ in~$G$.
\end{itemize}
In each case, we reach a contradiction. 

Thus, we have seen that, in case $F\not=\emptyset$, $G$ has no perfect matching cut separating~$A$ and~$B$, or $G$ contains the tree~$T$ as an induced subgraph. 
So, after at most~$|E|$ rounds, each for a candidate $ab\in E$ and in polynomial time, our algorithm will find out whether~$G$ has a perfect matching cut at all. 
In summary, we obtain:
\begin{theorem}\label{thm:poly}
\PMC\ is solvable in polynomial time for $T$-free graphs.
\end{theorem}

\subsection{Interval, chordal and pseudo-chordal graphs}
Recall that a graph has girth at least $g$ if and only if it has no induced cycles of length less than $g$. 
Thus, Theorem~\ref{thm:bip+deg3+girth} implies that \PMC\ remains hard when restricted to graphs without short induced cycles. 
This suggests studying \PMC\ restricted to graphs without long induced cycles, i.e., $k$-chordal graphs. Here, 
given an integer $k\ge 3$, a graph is \emph{$k$-chordal} if it has no induced cycles of length larger than $k$; the 3-chordal graphs are known as chordal graphs.

In this subsection we show that \PMC\ can be solved in polynomial time when restricted to what we call pseudo-chordal graphs, that contain the class of 3-chordal graphs and thus known to have unbounded mim-width~\cite{Kang}.

We begin with a concise characterization of interval graphs having
perfect matching cuts, to yield a polynomial-time algorithm deciding if an interval graph has a perfect matching cut which is much simpler than what we get by the mim-width approach~\cite{Bui}.   

\begin{fact}
\label{cliquefact}
Let $G$ have a vertex set $U\subseteq V(G)$ such that $G[U]$ is
connected
with every edge of $G[U]$ belonging to a triangle.
Then if $(X,Y)$ is a perfect matching cut of $G$
we must have $U \subseteq X$ or $U \subseteq Y$.
\end{fact}

This since otherwise we must have a triangle $K$ and two vertices $u, v$
with $u \in K \cap X$ and $v \in K \cap Y$ having a common neighbor
in $K$ so this cannot be a perfect matching cut.

If an interval graph $G$ has a cycle then it has a 3-clique. By Fact~\ref{cliquefact} 
these 3 vertices would have to belong to the same side
of the cut, and each would need to have a unique neighbor on the other
side of the cut. But then those 3 neighbors would form an asteroidal
triple, contradicting that $G$ was an interval graph. Thus
an interval graph which is not a tree does not have a perfect matching
cut.
A tree $T$ is an interval graph if and only if it does not have the
subdivided claw as a subgraph. Thus, $T$ is a caterpillar with basic
path $x_1,\ldots,x_k$, where $x_1$ and $x_k$ does not have a leaf attached,
while the other $x_i$ may have any number of leaves attached. 
(A \emph{caterpillar} is a tree with a (basic) path such that all vertices outside the path has a neighbor on the path.) 
If some $x_i$ has at least two leaves attached then $T$ does not have a
perfect matching cut $(X,Y)$, as $x_i \in X$ would imply that at least
one of those leaves is in $X$ and this leaf would not have a neighbor in
$Y$.
Since a leaf vertex and its neighbor must belong to opposite sides of
the cut, it is not hard to verify the following. 

\begin{fact}
An interval graph has a perfect matching cut if and only if it is a
caterpillar with basic path $x_1,\ldots,x_k$ such that any $x_i$ for $1 < i
< k$ has either zero or one leaf, and any maximal sub-path of
$x_1,\ldots,x_k$ with zero leaves contains an even number of vertices.
\end{fact}

In particular, caterpillars having a perfect matching cut can be recognized in polynomial time. 
For an arbitrary tree $T$ we can decide whether $T$ has a perfect matching as follows: 
Root $T$ at a vertex $r$ and let $r_1,\ldots, r_k$ be the children of~$r$. Then $T$ has a perfect matching cut if and only if there exists some $1\le i\le k$ such that each subtree $T_j$ rooted at $r_j$, $j\not=i$, has a perfect matching cut, and $T_i-r_i$ has a perfect matching cut. 
This fact implies a bottom-up dynamic programming to decide if $T$ has a perfect matching cut.   

A similar idea works for a large graph class that properly contains all chordal graphs. 
We will show a polynomial-time algorithm for what we call pseudo-chordal graphs. 
The maximal 2-connected subgraphs of a graph are called its blocks, and a block is non-trivial if it contains at least 3~vertices. 

\begin{definition}
A graph is \emph{pseudo-chordal} if, for every non-trivial block $B$,
every edge of $B$ belongs to a triangle.
\end{definition}

Note that chordal graphs are pseudo-chordal, but pseudo-chordal graphs
may contain induced cycles of any length, e.g. take a cycle and for any
two neighbors add a new vertex adjacent to both of them.

\begin{theorem}\label{thm:pseudo}
There is a polynomial-time algorithm deciding if a pseudo-chordal graph
$G$ has a perfect matching cut.
\end{theorem}
\begin{proof}
We first compute the blocks of $G$ and let $D$ be the
subgraph of $G$ formed by the edges of non-trivial blocks of $G$. 
Let $D_1, D_2,\ldots, D_k$ be the connected components of $D$. 
Note that by collapsing each $D_i$ into a supernode we can treat the
graph~$G$ as having a tree structure $T$ (related to the block
structure) with one node for each $v \in V(G) \setminus V(D)$, and a
supernode for each $D_i$. See Fig.~\ref{fig:Ex}. 
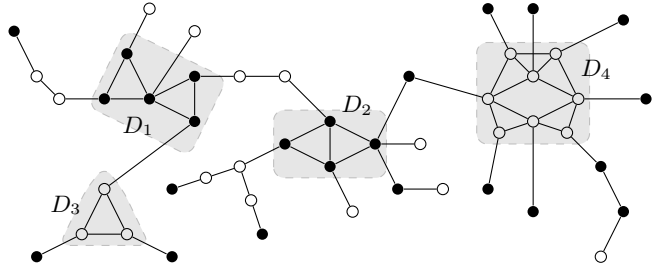
\begin{figure}[ht]
\tikzstyle{vertex}=[draw,circle,inner sep=1.5pt] 
\tikzstyle{vertexX}=[circle,inner sep=1.5pt,fill=black]
\tikzstyle{vertexY}=[draw,circle,inner sep=1.5pt,fill=lightgray]
\begin{center}
\begin{tikzpicture}[scale=.3] 
\node[vertexX] (1) at (1,11) {};
\node[vertex] (2) at (2,9) {};
\node[vertex] (3) at (3,8) {};
\node[vertex] (4) at (7,12) {};
\node[vertex] (5) at (9,11) {};
\node[vertex] (6) at (11,9) {};
\node[vertex] (7) at (13,9) {};
\node[vertexX] (8) at (18.5,9) {};
\node[vertex] (10) at (19,6) {};
\node[vertexX] (11) at (18,4) {};
\node[vertex] (12) at (20,4) {};
\node[vertex] (13) at (16,3) {};
\node[vertex] (14) at (11,5) {};
\node[vertex] (15) at (11.5,3.5) {};
\node[vertexX] (16) at (12,2) {};
\node[vertex] (17) at (9.5,4.5) {};
\node[vertexX] (18) at (8,4) {};
\node[vertexX] (21) at (2,1) {};
\node[vertexX] (22) at (8,1) {};
\node[vertexX] (23) at (22,4) {};
\node[vertexX] (24) at (24,3) {};
\node[vertexX] (25) at (27,5) {};
\node[vertexX] (26) at (28,3) {};
\node[vertex] (27) at (27,1) {};
\node[vertexX] (28) at (29,8) {};
\node[vertexX] (29) at (28,11.5) {};
\node[vertexX] (30) at (24,12) {};
\node[vertexX] (31) at (22,12) {};

\filldraw[rounded corners, fill=black!10!white, draw=lightgray, dashed, rotate around={-26:(7,8)}] (4.7,6.5) rectangle (9.7,10.3);
\node[vertexX] (x11) at (5,8) {};
\node[vertexX] (x12) at (6,10) {};
\node[vertexX] (x13) at (7,8) {};
\node[vertexX] (x14) at (9,9) {};
\node[vertexX] (x15) at (9,7) {};
\node at (6.5,6.8) {\small $D_1$};
\filldraw[rounded corners, fill=black!10!white, draw=lightgray, dashed] (12.5,4.5) rectangle (17.5,7.5);
\node[vertexX] (x21) at (13,6) {};
\node[vertexX] (x22) at (15,7) {};
\node[vertexX] (x23) at (15,5) {};
\node[vertexX] (x24) at (17,6) {};
\node at (16.2,7.7) {\small $D_2$};
\filldraw[rounded corners, fill=black!10!white, draw=lightgray, dashed] (3,1.5)--(7,1.5)--(5.8,4)--(5,5)--(4.2,4)--cycle; 
\node[vertex] (x31) at (5,4) {};
\node[vertex] (x32) at (6,2) {};
\node[vertex] (x33) at (4,2) {};
\node at (3.3,3.3) {\small $D_3$}; 
\filldraw[rounded corners, fill=black!10!white, draw=lightgray, dashed] (21.5,6) rectangle (26.5,10.5); 
\node[vertex] (x41) at (22,8) {};
\node[vertex] (x42) at (23,10) {};
\node[vertex] (x43) at (25,10) {};
\node[vertex] (x44) at (26,8) {};
\node[vertex] (x45) at (24,7) {};
\node[vertex] (a) at (24,9) {};
\node[vertex] (b) at (25.5,6.5) {};
\node[vertex] (c) at (22.5,6.5) {};
\node at (26.8,9.3) {\small $D_4$};

\draw (x11)--(x12)--(x13)--(x14)--(x15)--(x13)--(x11);
\draw (x22)--(x21)--(x23)--(x24)--(x22)--(x23);
\draw (x31)--(x32)--(x33)--(x31);
\draw (1)--(2)--(3)--(x11);
\draw (4)--(x12); \draw (5)--(x13); 
\draw (x14)--(6)--(7)--(x22); 
\draw (8)--(x24);
\draw (10)--(x24); \draw (12)--(11)--(x24); \draw (13)--(x23);
\draw (16)--(15)--(14)--(x21);
\draw (18)--(17)--(14);
\draw (x15)--(x31);
\draw (21)--(x33); \draw (22)--(x32);
\draw (x41)--(x42)--(x43)--(x44)--(x45)--(x41);
\draw (x41)--(a)--(x42); \draw (x43)--(a)--(x44);
\draw (x44)--(b)--(x45); \draw (x45)--(c)--(x41);
\draw (8)--(x41);
\draw (23)--(c); \draw (24)--(x45); \draw (b)--(25)--(26)--(27);
\draw (28)--(x44); \draw (29)--(x43); \draw (30)--(a); \draw (31)--(x42);  
\end{tikzpicture}
\caption{A pseudo-chordal graph and perfect matching cut given by $(X,Y)$
with $X$ being black vertices. Note the tree structure composed of (i)
those vertices that do not belong to a clique of size 3 and (ii) the
four supernodes $D_1,D_2,D_3,D_4$.}\label{fig:Ex}
\end{center}
\end{figure}

Note that since~$G$
is pseudo-chordal then by Fact~\ref{cliquefact} all the vertices in a
fixed supernode $D_i$ must be on the same side in any perfect matching
cut of $G$.

Our algorithm will pick a root $R$ of $T$ and proceed by bottom-up
dynamic programming on the rooted tree $T$. Each node $S$ of $T$ will be
viewed as the set of vertices it represents in $G$.
If $S$ is not the root of $T$ then we denote by $r(S)$ the unique vertex
of $S$ that has a parent in~$T$. For each node $S$ of $T$ we will
compute two boolean values that concern the subgraph~$G_S$ of $G$
induced by vertices of $G$ contained in the subtree of $T$ rooted at
$S$. These boolean values are defined as follows:

\begin{itemize}
\item $pmc(S)=\mathtt{true}$ if and only if $G_S$ has a perfect matching cut
\item $m(S)=\mathtt{true}$ if and only if $G_S \setminus r(S)$ has a perfect
matching cut where all vertices of $S \setminus r(S)$ are on the same
side.
\end{itemize}

We first initialize $pmc(S)$ and $m(S)$ to $\mathtt{false}$ for all nodes $S$ of
$T$. 
For a leaf $S$ of $T$ we set $m(S)=\mathtt{true}$ if $|S|=1$, i.e. if $S$ is not
a supernode.

Consider an inner node $S$ of $T$, with $S=\{v_1,\ldots,v_q\} \subseteq
V(G)$. In the rooted tree $T$, let the children of $S$ that contain a
neighbor of $v_i$ be $C(v_i)$ (note that each child of $S$ in $T$ has a
unique vertex that has a unique neighbor $v_i \in S$). Assuming the
values $pmc(\cdot)$ and~$m(\cdot)$ have been computed for all children of $S$, we do
the following:

\begin{itemize}
\item set $pmc(S)=\mathtt{true}$ if for each $v_i \in S$ we have
$C(v_i)=\{S_1,\ldots,S_k\}$ with $k \geq 1$ and we can find a child with
$m(S_i)=\mathtt{true}$ such that $pmc(S_j)=\mathtt{true}$ for the other $k-1$ children $j
\neq i$.
\item set $m(S)=\mathtt{true}$ if (i) for each $v_i \in S \setminus r(S)$ we have
$C(v_i)=\{S_1,\ldots,S_k\}$ with $k \geq 1$ and we can find a child with
$m(S_i)=\mathtt{true}$ such that $pmc(S_j)=\mathtt{true}$ for the other $k-1$ children $j
\neq i$, and (ii) for every child $S' \in C(r(S))$ we have
$pmc(S')=\mathtt{true}$.
\end{itemize}

For the root $R$ of $T$ we update $pmc(R)$ but not $m(R)$ since $r(R)$
is not defined. When we are done with the bottom-up dynamic programming,
then for the root $R$ of $T$ we note that $G=G_R$ so that by the
definition of the values $G$ has a perfect matching cut if and only if
$pmc(R)=\mathtt{true}$.

The correctness follows by structural induction on the tree $T$. By
definition, $pmc(S)=\mathtt{true}$ (respectively $m(S)=\mathtt{true}$) if and only if
there is a cut of $G_S$ so that every node (respectively every node
except $r(S)$) has a single neighbor, its \lq mate\rq, on the other side of
the cut. The values at leaves are initialized correctly according to
this definition. At an inner node $S$ we inductively assume the values
at children are correct and end up setting $pmc(S)$ to $\mathtt{true}$ if and only
if $G_S$ has a perfect matching cut, since for each node in $S$ we
require a single child neighbor that needs a mate, while all other child
neighbors are required to already have a mate. Similarly for $m(S)$ but
now all children of $r(S)$ are required to already have a mate. Since
each node $v$ of $S$ is a cut vertex of $G$ separating $G_S$ so that
each child of $S$ defines its own unique component, we can merge all the
cuts in all children while keeping all the nodes of $S$ on the same side
of the cut, to satisfy Fact~\ref{cliquefact} that requires all the nodes
in $S$ to be on the same side of the cut.
The runtime is clearly polynomial.\qed
\end{proof}

\section{Conclusion}\label{sec:conclusion}
We have shown that, assuming ETH, there is no $O^*(2^{o(n)})$-time algorithm for \PMC\ even when restricted to $n$-vertex bipartite graphs, and that \PMC\ remains $\NP$-complete when restricted to bipartite graphs of maximum degree~3 and arbitrarily large girth. 
This implies that \PMC\ remains $\NP$-complete when restricted to $H$-free graphs where $H$ is any fixed graph having a vertex of degree at least~4 or a cycle. This suggests the following problem for further research:
\begin{quote}
Let $F$ be a fixed forest with maximum degree at most~3. What is the computational complexity of \PMC\ restricted to $F$-free graphs?
\end{quote}
We have proved a first polynomial case for this problem where $F$ is a certain 6-vertex tree, including claw-free graphs and graphs without an induced 5-path. 

Our hardness result also suggests studying \PMC\ restricted to graphs without long induced cycles:
\begin{quote}
What is the computational complexity of \PMC\ on $k$-chordal graphs?
\end{quote}
It follows from our results that \PMC\ is polynomially solvable for 3-chordal graphs.

We have also given an exact branching algorithm for \PMC\ running in $O^*(1.2721^n)$ time. It is natural to ask whether the running time of the branching algorithm can be improved. 
Finally, as for matching cuts, also for perfect matching cuts it would be interesting to study counting and enumeration as well as FPT and kernelization algorithms.

\bibliography{perfectMC-arXiv}

\end{document}